\documentclass{article}

\PassOptionsToPackage{numbers,sort&compress}{natbib}
\usepackage[preprint]{neurips_2026}
\usepackage{rotating}

\usepackage[utf8]{inputenc} 
\usepackage[T1]{fontenc}    
\usepackage{hyperref}       
\usepackage{url}            
\usepackage{booktabs}       
\usepackage{amsfonts}       
\usepackage{nicefrac}       
\usepackage{microtype}      
\usepackage{xcolor}         

\usepackage{amsmath,amsfonts,bm}

\def\1{\bm{1}}

\DeclareMathAlphabet{\mathsfit}{\encodingdefault}{\sfdefault}{m}{sl}
\SetMathAlphabet{\mathsfit}{bold}{\encodingdefault}{\sfdefault}{bx}{n}

\DeclareMathOperator*{\argmax}{arg\,max}
\DeclareMathOperator*{\argmin}{arg\,min}

\usepackage[utf8]{inputenc} 
\usepackage[T1]{fontenc}    
\usepackage{url}            
\usepackage{booktabs}       
\usepackage{amsfonts}       
\usepackage{nicefrac}       
\usepackage{microtype}      
\usepackage{xcolor}         
\usepackage{color}
\usepackage{amsmath,amssymb,amsthm}
\usepackage{xspace}
\usepackage{mathtools}
\usepackage{dsfont}
\usepackage{tocvsec2}
\usepackage{wrapfig}
\usepackage{float}
\usepackage{array}

\usepackage{xcolor}         
\definecolor{mycyan}{RGB}{0,204,204}

\newcommand{\EE}{\mathbb{E}}

\newcommand{\blambda}{\bm{\lambda}}

\renewcommand{\epsilon}{\varepsilon}

\definecolor{ed}{RGB}{225,0,100}

\usepackage{hyperref}
\usepackage{url}

\usepackage{booktabs}
\usepackage{array}

\usepackage{lipsum}
\usepackage[utf8]{inputenc} 
\usepackage[T1]{fontenc}    
\usepackage{hyperref}       
\usepackage{url}            
\usepackage{booktabs}       
\usepackage{boldline}
\usepackage{cellspace}
\usepackage{amsfonts}       
\usepackage{nicefrac}       
\usepackage{microtype}      
\usepackage{amsmath,amsthm,amsfonts,amssymb,bm,accents,dsfont}

\usepackage{bbm}
\usepackage{graphicx}
\usepackage{style_files/mysymbol}
\usepackage{caption}
\usepackage{subcaption}
\usepackage{multirow}
\usepackage{blindtext}
\usepackage{multibib}
\newcites{supp}{Supplementary References}
\usepackage{tikz}
\usetikzlibrary{arrows,arrows.meta}
\usepackage{pgfplots}
\usepgfplotslibrary{groupplots}
\usepgfplotslibrary{fillbetween}
\definecolor{green}{rgb}{11, 82, 19}
\usepackage{hyperref}

\usepackage{algpseudocode,algorithm}

\newtheorem{proposition}{Proposition}

\newtheoremstyle{assume}
  {3pt}
  {3pt}
  {}
  {}
  {\bf}
  {.}
  { }
  {\thmname{#1}~\thmnumber{#2}~\thmnote{\textnormal{\textbf{(#3)}}}}
\theoremstyle{assume}
\newtheorem{assumption}{Assumption}

\newcommand{\terr}[1]{{\color{gray}\scriptsize $\pm$ #1}}

\newcommand{\norm}[1]{\ensuremath{\left\| #1 \right\|}}

\newcommand{\ubar}[1]{\ensuremath{\underaccent{\bar}{#1}}}
\def\subject{\text{s. to}}

\def\bbtheta{\mathbf{\theta}}

\newcommand{\eqrefp}[1]{(\hyperref[eqn_the_problem]{$\textup{P}_{#1}$})}
\newcommand{\eqrefpt}[1]{(\hyperref[eqn_variational]{$\textup{\~P}_{#1}$})}

\def\nnil{\nil}
\newcounter{prob}

\usepackage{style_files/custom_tex}
\usepackage{amsthm}
\usepackage{amsmath}
\usepackage{amssymb}
\usepackage{graphicx}
\usepackage{mathtools}
\usepackage{enumerate}
\usepackage{enumitem}
\usepackage{footnote}
\usepackage{float}
\usepackage{xspace}
\usepackage{multirow}
\usepackage{nicefrac}
\usepackage{wrapfig}
\usepackage{framed}
\usepackage{url}
\usepackage{units}
\usepackage{blkarray}

\usepackage{tikz}
\usetikzlibrary{arrows,chains,matrix,positioning,scopes}

\let\hat\widehat
\let\tilde\widetilde
\let\check\widecheck

\newcommand{\bx}{\boldsymbol{x}}

\newcommand{\bu}{\boldsymbol{u}}
\newcommand{\by}{\boldsymbol{y}}

\newcommand{\bg}{\boldsymbol{g}}

\newcommand{\bw}{\boldsymbol{w}}

\DeclareFontFamily{OMX}{MnSymbolE}{}
\DeclareFontShape{OMX}{MnSymbolE}{m}{n}{
    <-6>  MnSymbolE5
   <6-7>  MnSymbolE6
   <7-8>  MnSymbolE7
   <8-9>  MnSymbolE8
   <9-10> MnSymbolE9
  <10-12> MnSymbolE10
  <12->   MnSymbolE12}{}
\DeclareSymbolFont{mnlargesymbols}{OMX}{MnSymbolE}{m}{n}
\SetSymbolFont{mnlargesymbols}{bold}{OMX}{MnSymbolE}{b}{n}
\DeclareMathDelimiter{\llangle}{\mathopen}{mnlargesymbols}{'164}{mnlargesymbols}{'164}
\DeclareMathDelimiter{\rrangle}{\mathclose}{mnlargesymbols}{'171}{mnlargesymbols}{'171}

\usepackage{prettyref}

\newcommand{\savehyperref}[2]{\texorpdfstring{\hyperref[#1]{#2}}{#2}}
\newrefformat{eq}{\savehyperref{#1}{\textup{(\ref*{#1})}}}
\newrefformat{eqn}{\savehyperref{#1}{Equation~\eqref{#1}}}
\newrefformat{lem}{\savehyperref{#1}{Lemma~\ref*{#1}}}
\newrefformat{def}{\savehyperref{#1}{Definition~\ref*{#1}}}
\newrefformat{line}{\savehyperref{#1}{line~\ref*{#1}}}
\newrefformat{thm}{\savehyperref{#1}{Theorem~\ref*{#1}}}
\newrefformat{corr}{\savehyperref{#1}{Corollary~\ref*{#1}}}
\newrefformat{cor}{\savehyperref{#1}{Corollary~\ref*{#1}}}
\newrefformat{sec}{\savehyperref{#1}{Section~\ref*{#1}}}
\newrefformat{app}{\savehyperref{#1}{Appendix~\ref*{#1}}}
\newrefformat{ap}{\savehyperref{#1}{Appendix~\ref*{#1}}}
\newrefformat{assum}{\savehyperref{#1}{Assumption~\ref*{#1}}}
\newrefformat{as}{\savehyperref{#1}{Assumption~\ref*{#1}}}
\newrefformat{ex}{\savehyperref{#1}{Example~\ref*{#1}}}
\newrefformat{fig}{\savehyperref{#1}{Figure~\ref*{#1}}}
\newrefformat{alg}{\savehyperref{#1}{Algorithm~\ref*{#1}}}
\newrefformat{rem}{\savehyperref{#1}{Remark~\ref*{#1}}}
\newrefformat{conj}{\savehyperref{#1}{Conjecture~\ref*{#1}}}
\newrefformat{prop}{\savehyperref{#1}{Proposition~\ref*{#1}}}
\newrefformat{proto}{\savehyperref{#1}{Protocol~\ref*{#1}}}
\newrefformat{prob}{\savehyperref{#1}{Problem~\ref*{#1}}}
\newrefformat{claim}{\savehyperref{#1}{Claim~\ref*{#1}}}
\newrefformat{que}{\savehyperref{#1}{Question~\ref*{#1}}} 

\usepackage{tcolorbox}
\tcbuselibrary{listings,theorems}

\usepackage{xcolor}
\usepackage{bm}

\definecolor{mycyan}{RGB}{0,204,204}
\newcommand{\chiche}[1]{{\color{mycyan}[chiche: #1]}}
\definecolor{myorange}{RGB}{252, 105, 16}

\usepackage{booktabs}

\title{Every Sample Counts: Supervised Fine-Tuning of Language Models with Pointwise Constraints}

\author{%
 Ignacio Hounie\thanks{Equal contribution. Corresponding authors.}
 \\
 University of Pennsylvania \\
\texttt{ihounie@seas.upenn.edu}
\And
Ignacio Boero$^*$
 \\
 University of Pennsylvania \\
\texttt{iboero@seas.upenn.edu}
\And
 Alejandro Ribeiro
 \\
 University of Pennsylvania \\
\texttt{aribeiro@seas.upenn.edu}
  }

\begin{document}

\maketitle

\begin{abstract}
Fine-tuning language models often requires enforcing  constraints on individual inputs without compromising downstream performance. Existing constrained alignment methods impose constraints on average, which can induce undesirable disparities across inputs or users. We propose a novel alignment framework that addresses this gap by enforcing per-sample constraints while still minimizing an average loss. To mitigate the impact of overly restrictive constraints and outliers, we introduce a learned, sample-dependent relaxation that minimally adjusts the constraints, trading off a user-defined relaxation cost with the training objective. 
To address practical duality and optimization challenges, we develop an augmented Lagrangian approach tailored to this formulation.
We demonstrate the flexibility of the framework by instantiating it under distinct small language-model fine-tuning tasks and constraints: safety in instruction following, preferences in function calling and length in re-ranking. Across these settings, our approach reduces tail constraint violations while largely preserving the model's performance.
\end{abstract}


\section{Introduction}
Language Models (LMs) pre-trained on vast corpora of data have enabled unprecedented capabilities across a wide range of tasks. However, as these models are increasingly deployed in high-stakes domains~\citep{huang2024comprehensive, esposito2025largelanguagemodelsmissioncritical}, a distinction emerges between performance objectives-- what we \emph{want} a model to be as good as possible at -- and critical constraints that \emph{must} be enforced in deployment. As a result, 
a growing body of work uses constrained optimization frameworks to fine-tune LMs under requirements such as   safety~\citep{dai2024safe, huang2024one, wachi2024stepwise, liu2024enhancing-safe-dpo, peng2024enhancing}, length control~\citep{lagrange-llms}, robustness to adversarial attacks~\citep{ziegler2022adversarial} and avoidance of reward hacking~\citep{rlvr-with-multiple-rewards}.
%

Existing constrained methods focus only on satisfying requirements in expectation over the distribution of inputs. A key limitation of this approach is that constraint satisfaction for individual inputs is influenced both by their inherent difficulty and their representation in the training distribution. There is a growing body of empirical evidence showing that LLMs effectively underperform on~\emph{tail} data~\citep{llm-tails,code-tails, VLA_tails, hallucination-tails, bi2025gptoss-tail}. Consequently, models may satisfy constraints on average while still exhibiting significant violations on critical subsets of inputs, particularly those underrepresented in training data or more challenging to learn.
In this work, we address this limitation by proposing a novel pointwise constrained fine-tuning framework:
\begin{itemize}
    \item[(C1)] We enforce required properties as pointwise constraints over the input distribution, rather than only in expectation.
\end{itemize}
We tackle the resulting constrained problem through the dual domain. Although a growing body of work has studied constrained learning via dual formulations~\cite{chamon2020probably, bai2022achievingzeroconstraintviolation, 10480186, manolache2025learningapproximatelyequivariantnetworks, navid}, the non-convexity of neural networks prevents strong duality, weakening guarantees of both feasibility and optimality~\cite{elenter2024near}.
To overcome these optimization challenges, we turn to augmented Lagrangian methods, which are increasingly used in deep learning settings~\cite{al-fioretto-diffusion-robots, al-pinn, al-prunning, al-discrete-diffusion, ramirez2025dual, boero2025cole}, and enable strong duality and primal recoverability in a wider range of non-convex settings.
\begin{itemize}
    \item[(C2)] We apply Augmented Lagrangian methods for pointwise constrained fine-tuning.
\end{itemize}
%
%

Pointwise constraints can be overly restrictive, particularly when constraint losses conflict with the primary task objective, or in the presence of outliers. To address these challenges, we introduce a mechanism for learned pointwise constraint relaxations. 
This relaxation seeks an optimal trade-off between a user-defined relaxation cost and the training objective. 
\begin{itemize}
    \item[(C3)] We propose a mechanism that relaxes the constraints while jointly solving the learning task.
\end{itemize}
To showcase the flexibility of our approach,  we evaluate it across three small language-model finetuning tasks: instruction following, tool calling, and passage re-ranking. In these settings, we instantiate pointwise constraints that capture safety, preferences, and relevance, respectively. We empirically validate that pointwise constraints can reduce sample-level violations, and that controlling these tails can be beneficial.
\begin{itemize}
    \item[(C4)] We provide empirical evidence that our framework reduces the tails of the constraint-violation distribution, which translates to better trade-offs between the objective and constraint metrics.
\end{itemize}
To demonstrate tail reduction, we analyze the empirical distribution of constraint violations across methods. To evaluate performance, we compare objective--constraint trade-offs using task-specific downstream metrics. We show that our approach can expand the Pareto frontier in key areas, offering more favorable trade-offs than existing methods in most task-specific metrics. 
\section{Fine-tuning with Pointwise Constraints}~\label{sec:problem-form}
A LM $\pi_{\theta}(\cdot \,\vert\, \bx)$: $\mathcal{X} \to \Delta (\mathcal{Y})$ maps a token sequence $\bx = \{x_t\}_{t=1}^{|\bx|}  \in \mathcal{X}$ to a distribution over  outputs $\by \in \mathcal{Y}$. Given an objective function $\ell_0(\pi_\theta(\cdot \,\vert\, \bx), \by): \Delta (\mathcal{Y}) \times \mathcal{Y}  \to \mathbb{R}$, and $m$ constraint functions $\ell_i(\pi_\theta(\cdot \,\vert\, \bx), \by): \Delta (\mathcal{Y}) \times \mathcal{Y}  \to \mathbb{R}, \; i=1, \;...\;, m$  with tolerances $\epsilon_i$, we want to solve the following inequality constrained learning problem:
\begin{equation}\label{eqn:s-csft}\tag{P}
\begin{aligned}
   \pi_\theta^* = \argmin_{\theta \in \Theta}
         \; &\EE_{(\bx,\by)\sim\mathcal{D}_0}
        \Big[\ell_0(\pi_\theta(\cdot \,\vert\, \bx), \by)\Big]   \\[0.2cm]
    \text{s. to} \quad 
     &\ell_i (\pi_\theta(\cdot \,\vert\, \bx), \by) \leq \epsilon_i, \quad
    \mathcal{D}_i-\text{a.e.}
\end{aligned}
\end{equation}
This formulation allows us to optimize a given utility in expectation while enforcing critical requirements to hold pointwise.  In Section \ref{sec:dual} we show how to solve the problem using its dual formulation, where we assume $m=1$ and let the objective and constraint distributions be equal for notational conciseness. 
We first illustrate the relevance of formulation~\eqref{eqn:s-csft} through three LM fine-tuning tasks.

\paragraph{Instruction following and safety.} 
Safe alignment of LMs requires complying with user requests while reliably refusing \emph{all} harmful requests. 
We promote helpfulness in the LM answers by minimizing the KL-divergence against a helpful but unsafe reference model $\pi_{\text{ref}}$ over instruction pairs $(\bx_h, \by_h) \sim \mathcal{D}_\text{H}$. 
We ensure refusal to harmful requests by enforcing the probability that the LM gives a refusal answer $\by_r$ to be above a certain threshold $\epsilon_\text{U}$, for all unsafe prompts following a distribution $x_u \sim \mathcal{D}_{\text{U}}$. 
Moreover, to avoid the common but undesirable side effect of the LM refusing to answer safe prompts \cite{xstest}, we enforce the probability that the LM gives a refusal answer $\by_r$ to be below a certain threshold $\epsilon_\text{H}$, for all safe prompts following the marginal of the objective distribution $x_h \sim \mathcal{D}^X_{\text{H}}$. 
This yields the pointwise alignment problem:
\begin{equation}\label{prob:safety}
\begin{aligned}
\min_{\theta \in \Theta}. \quad & \EE_{(\bx_h,\by_h)\sim\mathcal{D}_{\text{H}}}
\Big[ D_{KL}\left(\pi_\theta(\by_h \mid \bx_h)\|\pi_{\text{ref}}(\by_h \mid \bx_h)\right) \Big] \\
\text{s.to} \quad & \pi_\theta(\by_r\mid \bx_u) \geq \epsilon_{\text{U}} \quad
x_u \sim \mathcal{D}_{\text{U}}-\text{a.e.}\\
& \pi_\theta(\by_r\mid \bx_h) \leq \epsilon_\text{H} \quad
x_h \sim \mathcal{D}^X_\text{H}-\text{a.e.}
\end{aligned}
\end{equation}
This differs from approaches that minimize a preference loss on a single dataset that contains pairs  of responses labeled for both their safety and helpfulness (e.g. \cite{kim2025safedpo}), or minimize negative log likelihood over a mixture of safe and helpful data (e.g. \cite{bianchi2023safetyllama}). See Appendix~\ref{app:related} for related work.

\paragraph{Tool Calling Preferences.} Language models often fail to identify \textit{when} to call a tool, leading to hallucinations when information is missing or correct tools are unavailable. We study a decision-making setting where the model must choose, for a given query $x$, one out of four possible behaviors: answering directly without calling tools, asking follow-up questions, alerting the user that needed tools are unavailable, or call a tool. We treat the expected KL divergence to a trained function-calling reference policy $\pi_{\text {ref }}$ as the minimization objective. We enforce sample-level constraints on the length-normalized log-probabilities: the correct behavior $y_w$ must meet a minimum threshold $\epsilon_{\text {win}}$, while it must fall below a tolerance $\epsilon_{\text{lose}}$ for all the incorrect behaviors ${y_l}_1,{y_l}_2, {y_l}_3$. The resulting pointwise constrained problem is:

\begin{equation}
\begin{aligned}
\min_{\theta \in \Theta}. & \;\mathbb{E}_{\bx \sim \mathcal{D}}\left[D_{\mathrm{KL}}\left(\pi_\theta(\by \mid \bx) \| \pi_{\mathrm{ref}}(\by \mid \bx)\right)\right] \\
\text { s. to } & \frac{1}{\left|\by_w\right|} \log \pi_\theta\left(\by_w \mid \bx\right) \geq \epsilon_{\text {win }}, \; \mathcal{D}-\text {a.e.} \\
& \frac{1}{\left|\by_l\right|} \log \pi_\theta\left({\by_l}_i \mid \bx\right) \leq \epsilon_{\text {lose}}, \; \; \mathcal{D}_i- \text {a.e.},\; i=1,\ldots,3
\end{aligned}
\end{equation}
Many preference optimization losses are penalties on the margin between preferred and rejected responses (e.g.~\cite{rafailov2024direct, meng2024simpo}). Others also add penalties to increase the preferred response probabilities (e.g.~\cite{xu2024cpo,Zhao2023SLiCHFSL}). See Appendix~\ref{app:related} for a discussion about related work.

\paragraph{Re-ranking and length optimization.}
When retrieving passages in RAG pipelines \cite{lewis2020retrieval}, it is often desirable to balance the relevance of a retrieved passage with its length, as processing long contexts increases latency and cost. We propose to train a re-ranking encoder $\pi_\theta$ that favors token efficiency (shorter passages) without sacrificing retrieval quality. We treat the preference for shorter passages as the minimization objective, using a length-biased ranking loss, which we denote as $\Lambda$. We enforce retrieval quality through a pointwise margin constraint. 
For each query $x$, positive passage $y_p$, and negative set $\mathbf{y_n}=\{y_n^i\}$ drawn from $\mathcal D$, the margin between the positive score $\pi_\theta(y_p\mid x)$ and each negative score $\pi_\theta(y_n^i\mid x)$ must be at least $\epsilon$. The resulting optimization is:

\begin{equation}
\begin{aligned}
\min_{\theta \in \Theta}. \quad & \EE_{(x,y_p,\mathbf{y_n}) \sim \mathcal{D}}
\Big[ \Lambda\big(\pi_\theta({y_p} \mid x), \pi_\theta({\mathbf{y_n} \mid x)}\big) \Big] \\
\text{s. to} \quad &  \pi_\theta(y_p \mid x) - \pi_\theta(\mathbf{y_n} \mid x) \geq \epsilon, \quad
 \mathcal{D}-\text{a.e.}
\end{aligned}
\end{equation}

To the best of our knowledge, length-aware re-ranking is a novel task. We include this example to illustrate how our framework can accommodate secondary goals without compromising performance.

\section{Dual Pointwise Fine-tuning}
\label{sec:dual}
Let $\blambda:\mathcal{X}\times\mathcal{Y}\to \mathbb{R}_+$ be a nonnegative function in $L^1_+(\mathcal{D})$. The corresponding Lagrangian is
\begin{align}\label{eq:lagrangian}
    \mathcal{L}(\pi_\theta,\blambda)
    \;:=\;
    \mathbb{E}\;\![\ell_0\big(\bbx, \bby, \pi_\theta\big)\;+\blambda(\bx,\bby)\,\big(\ell(\bx,\bby,\pi_\theta)-\epsilon\big)],   
\end{align}
 
where the Lagrangian minimizer is defined as 
\begin{align}
\pi_{\theta}(\blambda)\in \argmin_{\theta \in \Theta} \mathcal{L}(\pi_\theta,\blambda).
\end{align}
Finding $\pi_{\theta}(\blambda)$ amounts to solving an unconstrained problem, where the constraint at each point $(\bbx,\bby)$ is weighted by $\blambda(\bbx,\bby)$. Its value is known as the dual function; $g(\blambda) = \mathcal{L}(\pi_{\theta}(\blambda),\blambda)$.
Maximizing the dual function over the multipliers $\blambda$ yields the \emph{dual problem} of problem \ref{eqn:s-csft},
\begin{align}\label{eq:dual-parameterized}\tag{D}
\blambda^\star
    \;\in\;
    \argmax_{\blambda\in L^1_+(\mathcal{D})} g(\blambda).
\end{align}

In convex optimization problems, the dual formulation introduces an alternative strategy to solve constrained problems, as the primal solution is a Lagrangian minimizer for the optimal dual variable. However, for most LM parametrizations, such as transformers, the problem \ref{eqn:s-csft} is not convex with respect to the parameters $\theta$. Therefore, $\pi_{\theta}(\blambda^\star)$ is not necessarily related to a solution of problem \ref{eqn:s-csft}. We overcome these challenges by turning to Augmented Lagrangian methods, presented next.

\subsection{Augmented Lagrangian} Augmented Lagrangian methods for inequality-constrained problems were proposed to close the duality gap in nonconvex problems~\cite{RockaDual}; see also~\citep{AL-review}. This consists of replacing the linear dependence between $\blambda$ and $(\ell - \epsilon)$ in \eqref{eq:lagrangian} with an augmented, often quadratic, one. In particular, using the augmenting function introduced in~\cite{RockaDual}, the Augmented Lagrangian of problem~\eqref{eqn:s-csft} is:

\begin{equation}
\begin{aligned}\label{eq:aug_lagrangian_og}
\mathcal{L}_{A}(\pi_\theta,\blambda, \alpha) =  \mathbb{E}\Big[\ell_0(\bx, \bby, \pi_\theta)  \; + \alpha \Big(\ell(\bx, \by, \pi_\theta)-\epsilon+\frac{\lambda(\bx, \by)}{2\alpha}\Big)_+^2- \frac{\lambda(\bx, \by)^2}{4\alpha}\Big],
\end{aligned}
\end{equation}

where $(x)_+ = \text{max}\{0,x\}$ clips negative values to $0$. Evaluating \eqref{eq:aug_lagrangian_og} on the Augmented Lagrangian minimizer,

\begin{align}
\pi_{\theta}^A(\blambda, \alpha)\in \argmin_{\theta \in \Theta} \mathcal{L}_A(\pi_\theta,\blambda,\alpha),
\end{align}

yields the augmented dual function $g_{A}(\blambda, \alpha)= \mathcal{L}_{A}(\pi_\theta^A(\blambda,\alpha),\blambda,\alpha)$. Analogous to the dual problem, maximizing $g_{A}(\lambda,\alpha)$ yields the \textit{augmented dual problem},

\begin{equation}\label{prob:aug_dual}\tag{$\text{D}_\text{A}$}
(\blambda_A^\star, \alpha^*_A)
    \;\in\;
    \argmax_{\blambda\in L^1_+(\mathcal{D}), \; \alpha \in \mathbb{R}_+} g_{A}(\blambda, \alpha).
\end{equation}

The standard and augmented dual problems differ in how constraint violations are weighted during the primal minimization. 
The standard Lagrangian uses a linear dependence,  which may be unable to account for the adverse curvature present in nonconvex problems. The quadratic term in the augmented Lagrangian provides an additional curvature correction, enabling guarantees beyond those of the standard dual. 
In particular, under the conditions studied in~\cite{boero2025cole}, an approximate solution for a constrained learning problem can be obtained via the primal minimizer of the empirical augmented Lagrangian at the optimal augmented dual pair~\cite[Theorem 2.3]{boero2025cole}. Empirically, we also find that the augmented formulation outperforms the standard Lagrangian in our settings; see Appendix~\ref{apx:aug_dual_vs_dual}.

\subsection{Regularization in the Dual Domain}\label{sec:resilience}
By imposing constraints to hold almost everywhere, one challenge that might render problem~\eqref{eqn:s-csft} impractical is that the constraint level $\epsilon$ may be overly restrictive in some low probability regions. A possible approach  is to relax the constraint to hold with high probability, instead of almost everywhere. However, this approach treats all violations equally regardless of their magnitude, allowing arbitrarily large violations as long as they are restrained to a low probability region. Instead, we propose a relaxation that controls both the frequency and severity of the violations.

Let $\bu \in L_+^\infty(\mathcal{D})$ be a pointwise relaxation function and let $c:L_+^\infty(\mathcal{D})\to\reals$ be a convex cost that measures the severity of the violations. We define the \textit{relaxed primal problem} as:
\begin{equation}\label{eqn_res_primal}\tag{$\text{P}_\text{R}$}
\begin{aligned}
(\pi_\theta^R, \bu_+^R) \in \argmin_{\theta \in \Theta, \, \bu \in L^\infty(\mathcal{D})}
         \;  \;
          &\EE
        \Big[\ell_0(\bx, \bby, \pi_\theta) \Big] + c(\bu)
        \\[0.1cm]
        \subject \; \;
     &\ell(\bx, \by, \pi_\theta) \, 
        \;\leq\;
        \epsilon + \bu(\bx,\by) \ \quad   \mathcal{D}-a.e.
\end{aligned}
\end{equation}
The function $\bu$ relaxes~\eqref{eqn:s-csft} by increasing pointwise the admissible threshold from $\epsilon$ to $\epsilon+\bu(\bx,\by)$. The cost $c(\bu)$ penalizes these relaxations, which can account for both frequency and severity. In this sense, problem~\eqref{eqn_res_primal} selects the optimal relaxation according to the cost $c$. Soft-margin SVMs are a classical finite-dimensional example; see Appendix~\ref{app:related}. In the context of constrained learning with average constraints, this has been called a \emph{Resilient} relaxation due to its robustness to mis-specification~\cite{hounie2024resilient}. 
Relaxing a constrained problem and penalizing the relaxation has a well-known dual-domain interpretation.  It corresponds to regularizing the dual function by the Fenchel conjugate of the function $c$; see~\cite{rockafellar1997convex}. In particular, the choice of $c(u)= \beta ||\bu||^2 := \beta E[\bu^2(\bx,\by)]$ induces a quadratic penalty on the dual multipliers. 

\begin{proposition}~\label{prop:resilient_augmented_dual}
The augmented dual problem of ~\eqref{eqn_res_primal} with relaxation cost $c(\bbu) = \beta ||\bu||^2,\; \beta>0$, is
\begin{align}\label{eq:res_aug_lagrangian_dual}\tag{$\text{D}_\text{R}$}
    (\blambda^*_R, \alpha^*_R) = \argmax_{\blambda \in L^1_+(\mathcal{D}), \alpha \in \reals_+} \; \Big[ g_R(\blambda,\alpha) := 
    \;  g_A\left(\frac{\blambda \beta}{\alpha+\beta}, \frac{\alpha \beta}{\alpha+\beta}\right) - \frac{||\blambda||^2}{4\beta(\alpha+\beta)} \Big].
\end{align}
\begin{proof}
    See Appendix~\ref{apx:proof-res}
\end{proof}

\end{proposition}

Proposition~\ref{prop:resilient_augmented_dual} provides a practical way to solve the relaxed problem: the variable $\bu$ does not need to be optimized directly, since its effect is absorbed into the dual problem through a regularization term and a rescaling of the dual variables. As the appropriate level of relaxation is generally not known in advance, $\beta$ acts as a tunable parameter controlling how expensive violations are. Smaller values of $\beta$ allow larger relaxations and may improve the objective value, whereas larger values enforce the constraints more strictly. In particular, the original constrained problem is recovered in the limit $\beta\to\infty$: in the primal domain, any nonzero relaxation receives infinite cost and $\bu=0$ is enforced; equivalently, in the dual domain, the quadratic regularization vanishes and the rescaling of the augmented-dual variables disappears. A sensitivity-based characterization of the optimal relaxation is given in Appendix~\ref{apx:dual_reg_apx}. In Section \ref{sec:exps} we show  empirically that varying $\beta$ enables efficient control of the tradeoff between objective value and constraint satisfaction.

\subsection{Empirical Dual Gradient Ascent}~\label{sec:algo} The augmented Lagrangian in~\eqref{eq:aug_lagrangian_og} still depends on the unknown distribution $\mathcal{D}$. 
Following standard empirical risk minimization, we replace the expectation with the sample average over an i.i.d. dataset $S=\{(\bbx_i,\bby_i)\}_{i=1}^N$ drawn from $\mathcal{D}$. 
We also replace the functional multiplier $\blambda(\bx,\by)$ with sample-wise multipliers $\blambda\in\mathbb{R}^N_+$, assigning one nonnegative multiplier to each training sample. 
The empirical augmented Lagrangian $\hat{\mathcal{L}}_{A}(\pi_\theta,\blambda):\Theta\times\mathbb{R}^N_+\to\mathbb{R}$ is then defined as

\begin{equation}
\begin{aligned}\label{eq:aug_lagrangian_emp}
\hat {\mathcal{L}_{A}}(\pi_\theta,\blambda, \alpha) =  \frac{1}{N}\sum_1^N \Big[\ell_0(\bx_i, \bby_i, \pi_\theta)  \; + \alpha \Big(\ell(\bx_i, \by_i, \pi_\theta)-\epsilon+\frac{\lambda_i}{2\alpha}\Big)_+^2- \frac{\lambda_i^2}{4\alpha}\Big],
\end{aligned}
\end{equation}
 
with the respective primal minimizer
\begin{align}\label{eq:primal_min_emp}
\hat \pi_{\theta}^A(\blambda, \alpha)\in \argmin_{\theta \in \Theta} \hat{\mathcal{L}}_A(\pi_\theta,\blambda,\alpha).
\end{align}

For fixed values of $\blambda$ and $\alpha$, the empirical Lagrangian in~\eqref{eq:aug_lagrangian_emp} is simply a sample-averaged training loss, where each sample contribution is reweighted and reshaped according to its constraint violation and multiplier $\lambda_i$. Thus, solving~\eqref{eq:primal_min_emp} reduces to the usual task of minimizing a differentiable empirical objective over the model parameters, which can be approached with the same stochastic-gradient methods used for standard empirical risk minimization. Then, we can define the empirical approximation of the augmented dual function by evaluating the empirical Lagrangian at this minimizer $\hat g_{A}(\blambda, \alpha)=  \hat{\mathcal{L}_{A}}(\hat \pi_{\theta}^A(\blambda, \alpha),\blambda, \alpha)$. Maximizing it gives the \textit{empirical augmented dual problem}:
\begin{align}\label{eq:emp-dual-aug}
    (\hat{\blambda}_A^\star,\hat\alpha_A^\star)
    \in
    \argmax_{\blambda\in \mathbb{R}^N_+,\,\alpha>0}
    \hat g_A(\blambda,\alpha).
\end{align}

Using the same construction, the \textit{empirical relaxed dual problem} is given by
\begin{align}\label{eq:emp-dual-res}
(\hat{\blambda}_R^\star,\hat\alpha_R^\star)
\in
\argmax_{\blambda\in\mathbb{R}^N_+,\,\alpha>0}
\Big[ \hat g_R(\blambda,\alpha) := 
    \;  \hat g_A\left(\frac{\blambda \beta}{\alpha+\beta}, \frac{\alpha \beta}{\alpha+\beta}\right) - \frac{\sum_{i=1}^N \blambda_i^2}{N4\beta(\alpha+\beta)}\Big].
\end{align}

To solve problems~\eqref{eq:emp-dual-aug} and~\eqref{eq:emp-dual-res}, we use dual gradient ascent over $\blambda$ with a fixed augmentation parameter $\alpha$. 
This corresponds to the \textit{shifted penalty} method, which converges to an optimal dual pair when $\alpha$ is sufficiently large~\cite[Algorithm 7.3]{shifted}. Although the theoretical threshold for a sufficiently large $\alpha$ is unknown a priori, we find empirically that performance is relatively insensitive to this choice; see Appendix~\ref{apx:alpha_stab}.

Computing gradients of $\hat g_A(\blambda,\alpha)$ and $\hat g_R(\blambda,\alpha)$ with respect to $\blambda$ requires evaluating the primal minimizer $\hat\pi_\theta^A(\blambda,\alpha)$ defined in~\eqref{eq:primal_min_emp}. 
Since solving this inner minimization exactly at every dual iteration is prohibitively expensive, we approximate it with a small number of primal gradient steps per dual ascent step, resulting in Algorithm~\ref{alg:primal-dual}. 
This inexact update is standard in primal-dual methods: convergence can still be guaranteed as long as the primal minimization error decreases sufficiently across iterations~\citep[Theorem~3.1]{boero2025cole}. 
We defer further details on the shifted-penalty algorithm, the empirical dual gradients, and the implementation to Appendix~\ref{app:algorithm}.

Another advantage of relaxing the empirical dual problem~\eqref{eq:emp-dual-res} is that, under mild conditions, it has zero duality gap, as shown in Appendix~\ref{app:Theo}.
Developing finite-sample bounds on the convergence of the empirical dual problem to its population counterpart is beyond the scope of this work. Nevertheless, in Section~\ref{sec:exps}, we empirically observe that both the objective and the distribution of constraint violations generalize out of sample.


\begin{table*}[t]
\caption{Comparison of constraint enforcement and resulting multipliers.
Pointwise constraints induce a distinct multiplier per sample in the dual,
whereas the average-constraint dual uses a single scalar multiplier; the penalty uses
a fixed weight.}
\centering
\setlength{\tabcolsep}{6pt}
\renewcommand{\arraystretch}{1.2}
\begin{tabular}{llll}
\toprule
\textbf{Formulation} 
& \textbf{Constraint} 
& \textbf{Multiplier} \\
\midrule
\textbf{pen}
& $-$
& $\lambda_0$ (fixed) \\[1mm]

\textbf{avg}
& $\EE\!\big[\ell(\bx_i,\by_i,\pi_\theta)\big] \le \epsilon$
& $\lambda$ (variable)\\[1mm]

\textbf{point} 
& $\ell(\bx_n,\by_n,\pi_\theta) \le \epsilon $
& $\blambda(\bx_n,\by_n)$ \\[1mm]

\textbf{relax}
& $\ell(\bx_n,\by_n,\pi_\theta) \le \epsilon + \bu(\bx_n,\by_n)$
& $\blambda(\bx_n,\by_n)$ \\[1mm]

\bottomrule
\end{tabular}
\label{tab:formulations_multipliers}
\end{table*}
\section{Numerical Experiments}\label{sec:exps}

In this section, we evaluate whether pointwise constraints provide a practical advantage over average-constrained or unconstrained formulations. Satisfying constraints only on average can still allow large violations on individual samples. Explicitly, the dual problem of average constraints corresponds to a single constant multiplier across samples, replacing  sample-wise $\blambda(\bx,\by)\in L^q_+(\mathcal D)$ with a scalar $\lambda\in\mathbb R_+$; see Appendix~\ref{sec:appx-avg-form}. Fixed penalties are even more restrictive, since this scalar is chosen in advance rather than optimized. Both approaches ignore heterogeneity across samples, applying the same penalty to inputs that may require different levels of enforcement. 
Since pointwise multipliers adapt the penalty at the sample level, we expect them to reduce the lower tail of the constraint-violation distribution. This motivates our first hypothesis to validate:
\begin{itemize}
    \item[(\textbf{H1}) ] Pointwise constraints reduce the frequency and severity of sample-level violations.
\end{itemize}

%
%
However, reducing tail violations is only useful if it also improves task-level behavior. We therefore also validate the second hypothesis: 

\begin{itemize}
    \item[(\textbf{H2}) ] 
    Tail reduction translates to better tradeoff between objective and constraints downstream metrics.
\end{itemize}

To verify the first hypothesis, we compare four formulations across all experiments: pointwise augmented dual learning (\textbf{point}), its relaxed variant (\textbf{relax}), average-constraints (\textbf{avg}), and a fixed-penalty (\textbf{pen}); see Table~\ref{tab:formulations_multipliers}. 
To evaluate the second hypothesis, we additionally compare against task-specific baselines for each setting. 

We focus on settings involving small language models (i.e., less than $7$ billion parameters), whose lower computational requirements enhance reproducibility~\cite{semmelrock2025reproducibility} and are a cost effective alternative in many tasks~\cite{belcak2025smalllanguagemodelsfuture, wang2025comprehensive}. All experiments were run using a fractional GPU allocation equivalent to one quarter of a single NVIDIA B200 GPU, which corresponds to $45$GB of GPU memory. Runtime and memory overhead are discussed in Appendix~\ref{apx:ablations}. Next, we briefly describe the experimental settings for each task and describe them thoroughly in Appendix~\ref{app:appendix-exp-details}.


\textbf{Tool Calling with Hallucination Mitigation}.
We use the \texttt{When2call} dataset~\cite{when2call}, derived from BFCL~\cite{bfcl} and APIGen~\cite{apigen}, to evaluate tool-calling decision-making: the model selects between tool execution, follow-up questions, or refusing unanswerable queries. The data includes $4.5$k tool-calling, $3$k follow-up, and $3$k unable to answer training examples, from which we hold out a random $5$\% subset for validation. We use Llama-3.2-1B-Instruct~\cite{grattafiori2024llama3} and xLAM-2-1b-fc-r~\cite{prabhakar2025apigen-xlam} as base models. Following~\cite{when2call}, we evaluate behavior selection using an offline multiple-choice format.
 
\textbf{Instruction Following with Safety Refusal}. The helpfulness dataset is Alpacalong-1k~\cite{alpaca-long1k}, a subset of the longest Alpaca instruction--response pairs~\cite{alpaca}. Safety data is a filtered 12k samples subset of BeaverTails~\cite{ji2024beavertails}. We also use the evaluation split of the PKU-SafeRLHF-30k dataset to evaluate safety. We use a Llama-7b~\cite{touvron2023llama} base model finetuned on Alpacalong-1k as a base model, and safeDPO~\cite{kim2025safedpo} as a baseline. We evaluate safety by generating responses to unsafe prompts and scoring them with a learned harmfulness classifier as in~\cite{dai2024safe} and rule-based keyword matching as in~\cite{bianchi2023safetyllama}; helpfulness is measured with AlpacaEval~\cite{alpaca_eval} on benign instructions.

\textbf{Re-ranking with Length Optimization}. We train and evaluate on a portion of {\emph{MS MARCO} v2.1} \cite{bajaj2016ms} with $50$k training and $5$k validation instances, each consisting on a query, a positive passage and $9$ negative passages. We use ModernBert-base~\cite{mordern-bert} as the base model.
For relevance we report {MRR@10} and {Hit@3}. For length efficiency we report {AvgLength@3} and a length-weighted ranking metric which we call {LenRank@10}; see equation \eqref{apx:len-rank}. We compare against training \textit{monoBERT} \cite{nogueira2019multi} on our dataset, and out-of-the-box rerankers \textit{BGE-large }\cite{bge_embedding} and \textit{ms-marco-MiniLM} \cite{hf_cross_encoder_msmarco_minilm_l12_v2}.
%

\subsection{Constraint distributions and objective value} 

\begin{table}[b]
\caption{Statistics of objective and constraint violations in the test set across reranking, function calling and safety  settings. The Conditional Value at Risk (CVaR) represents the average loss for samples exceeding the 95-th percentile of the loss distribution. ~\emph{Our approach results in smaller tail violations with lower average objectives.}}
\label{tab:eval_slacks_stats}
\centering
\begin{tabular}{lllllll}
\hline
 & \multicolumn{2}{c}{\textbf{Re-ranking}} & \multicolumn{2}{c}{\textbf{Function Calling}} & \multicolumn{2}{c}{\textbf{Safety}} \\
\cline{2-3}\cline{4-5}\cline{6-7}
 & \textbf{Constraint} & \textbf{Objective} & \textbf{Constraint} & \textbf{Objective} & \textbf{Constraint} & \textbf{Objective} \\
 & \textbf{CVaR} & \textbf{Mean} & \textbf{CVaR} & \textbf{Mean} & \textbf{CVaR} & \textbf{Mean} \\
\hline
pen   &  105 \terr{10}   &   33 \terr{5}    &   16 \terr{5}    &   8.16 \terr{0.04}  &  0.32 \terr{0.03}  &  1.9 \terr{0.8} \\
avg   &  120 \terr{23}   &   42 \terr{2}    &   18 \terr{2}    &   7 \terr{1}       &  0.31 \terr{0.04}  &  0.38 \terr{0.03}\\
point &   2.4 \terr{0.5}  &   8.2 \terr{0.02} &   9.4 \terr{0.8} &   7.5 \terr{0.5}   &  0.09 \terr{0.002} &  0.19 \terr{0.06}\\
\hline
\end{tabular}
\end{table}

\textbf{Pointwise constraints reduce large violations.} We begin by analyzing the distribution of constraints across samples for our method, average constraints and fixed penalties. Figure~\ref{fig:CDF-reranking} shows the empirical Cumulative Distribution Function (CDF) for the reranking constraint evaluated on test samples. For constraints that are satisfied by a large margin, our methods CDF lies below all others, indicating that competing approaches have a higher proportion of samples with smaller losses. Conversely, for larger constraint violations, our approach's CDF rises substantially above all others, showing that it has very few samples with large violations. This suggests that average constraints and fixed penalties satisfy by large margins the constraints on ``easy'' samples, but largely fail on ``hard'' ones. In contrast, our approach ensures more consistent performance among samples. 
We include CDF plots for all other tasks in Appendix~\ref{app:results}, where similar trends are observed.

\begin{figure}[t]
    \centering
    \includegraphics[width=0.95\linewidth]{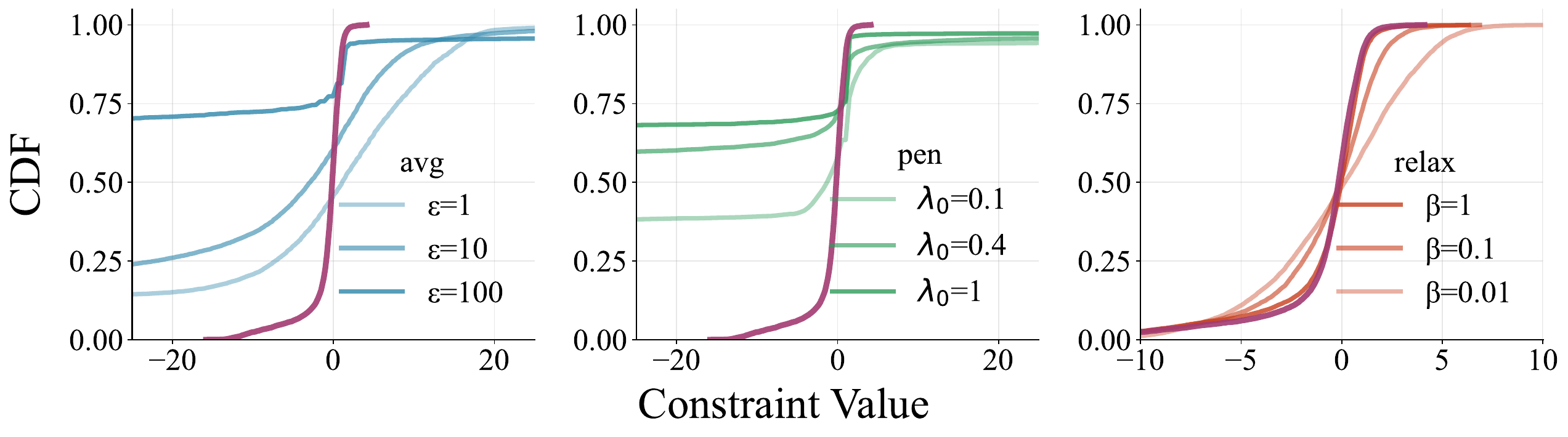}
  \caption{Empirical Cumulative Distribution Function (CDF) of sample-wise constraint values $\ell(\pi_{\theta},x,y)$ in the test set for the reranking task. Positive values violate the constraint. We compare our method (purple line) with average constraints (left), penalty methods (center) and relaxed pointwise constraints (right). 
  \emph{Our approach results in significantly fewer samples incurring large violations.}}
  \label{fig:CDF-reranking}
\end{figure}

\textbf{Tighter average constraints do not reduce large violations.}
For a given constraint level, imposing pointwise constraints is more restrictive than imposing the constraint only on average.
Therefore, we analyze whether tightening the average constraint level could lead to similar tail behaviour, which would require the constraint violations to be sufficiently concentrated around its mean. As shown in Figure~\ref{fig:CDF-reranking} (right), we observe that the distribution of the constraint remains long tailed even when tightening the average constraint level (decreasing $\epsilon$). That is, stricter average constraints do not improve larger losses (to the right of the plot).
%

\textbf{Larger fixed penalties do not reduce large violations.} Similarly, we analyze the effect of increasing the fixed penalty $\lambda_0$ on the distribution of constraint violations. As illustrated in Figure~\ref{fig:CDF-reranking}(center), although larger penalties shift the distribution toward smaller average violations, the right tail remains largely unchanged. Explicitly, while moderate violations are reduced, large violations persist.

\paragraph{Pointwise constraints achieve better objective trade-offs.} 
While bounding the constraints for every sample, solving the pointwise constrained problem requires attaining the best average objective. In Table~\ref{tab:eval_slacks_stats} we show that across the three settings pointwise constraints achieve a smaller average objective than average constraints, while attaining smaller violation tail statistics (Conditional Value at Risk).
This illustrates that a constant $\blambda$ over-penalizes inputs, resulting in a worse average objective.

\begin{figure}[b]
        \centering        \includegraphics[width=0.95\linewidth]{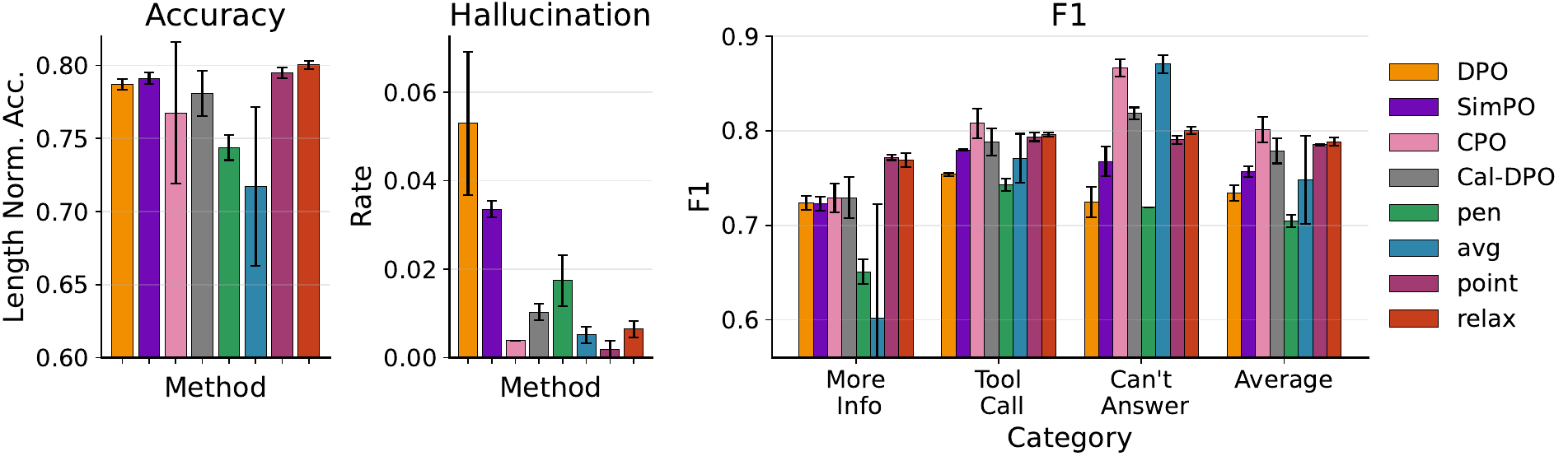}
    \caption{Length normalised accuracy, Hallucination Rates and F1 scores for Llama-3.2-1B-Instruct~\cite{grattafiori2024llama3} finetuned on the when2call~\cite{when2call} dataset for different methods. Height of bars denotes the mean and errorbars  the standard deviation, computed across three seeds.  A full table with results for other models is included in Table~\ref{tab:when2call_results}. \emph{Pointwise constraints achieve a more homogeneous distribution of errors across different categories.}}
    \label{fig:comp_methods-tool}
\end{figure}
\textbf{Relaxation controls the objective  vs. constraint trade-off.}
The relaxation cost function controls the trade-off between relaxing the constraints and the average objective. To illustrate this, we perform an ablation on the hyperparameter $\beta$ in the quadratic relaxation cost $c(\bbu) = \beta \|\bbu\|_2^2$. 
As $\beta$ increases, the relaxations decrease and the solution approaches its hard constrained counterpart, as depicted in Figure~\ref{fig:CDF-reranking} (right). Smaller values of $\beta$ lead to improved objective (Figure~\ref{fig:pareto_reranker}).



\subsection{Implications in downstream performance}

\paragraph{Improvements across function calling error categories.} Relying solely on aggregate metrics 
can obscure differences between alignment methods. For example,  while SimPO~\cite{meng2024simpo} and Pointwise constraints achieve similar length-normalized accuracies (Figure~\ref{fig:comp_methods-tool} left plot), this masks significant performance disparities across different  error types and output categories (Figure~\ref{fig:comp_methods-tool}, middle and right plot). 
Imposing Average Constraints or using the CPO loss 
skews the F1 score distribution across error categories, improving on the ``Unable to answer'' category 
  at the expense of ``Follow-up question'' performance. 
   In contrast, Pointwise constraints yield a more uniform improvement of performance across both categories. 
   Finally, the Resilient relaxation matches or improves performance with respect to the pointwise method across all prompt types while maintaining near-zero hallucination rates. We observe similar trends across XLaM models, as detailed in Appendix~\ref{app:results-func} (Table~\ref{tab:when2call_results}). 
\begin{figure}[t]
    \centering
    \begin{minipage}[c]{0.45\linewidth}
        \centering
        \includegraphics[width=\linewidth]{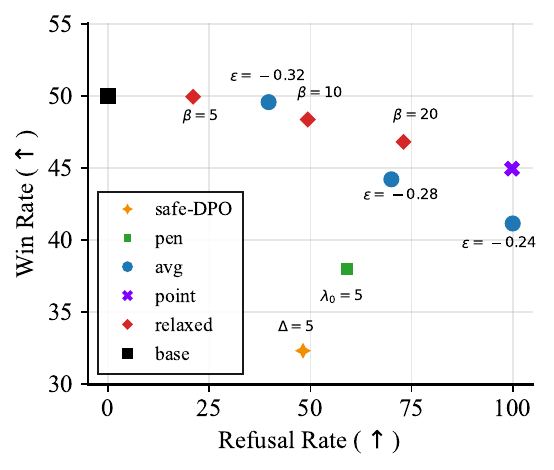}
        \end{minipage}
    \begin{minipage}[c]{0.5\linewidth}
        \centering
\begin{tabular}{lcc}
\toprule 
& \multicolumn{1}{c}{Alpaca} & SafeRLHF \\ 
& $D_{\mathrm{KL}}$($\downarrow$) & cost($\uparrow$) \\
\midrule
Safe-DPO ($\Delta = $ 5) & 1.4\terr{0.2} & 45\terr{20} \\
Safe-DPO ($\Delta = $ 10) & 1\terr{1} & 47\terr{30} \\
pen ($\lambda_0 = $ 5) & 1.9\terr{0.8} & 71\terr{9} \\
avg ($\epsilon = - $ 0.32) & 0.37\terr{0.04} & 53\terr{10} \\
avg ($\epsilon = - $ 0.24) & 5\terr{1} & 78\terr{30} \\
point & 0.2\terr{0.08} & 96\terr{2} \\
relax ($\beta = $ 20) & 0.12\terr{0.01} & 77\terr{2} \\
relax ($\beta = $ 10) & 0.087\terr{6e-6} & 60\terr{1} \\
\bottomrule
\end{tabular}
    \end{minipage}
    \caption{Safety vs Helpfulness trade-off. (Left) length controlled win rates on AlpacaEval~\cite{alpaca_eval} and refusal rates over harmful held-out prompts from beavertails~\cite{ji2024beavertails}. (Right) KL divergence in a held out set of alpaca~\cite{alpaca}, and percentage of responses classified as harmless ($\text{cost}<0$) by \texttt{beaver-7b-unified-cost} from PKU-SafeRLHF-30k~\cite{dai2024safe} evaluation split. This is a subset of runs to ease readability, Table~\ref{tab:app-safety} includes results for a larger hyperparameter sweep. \emph{Pointwise constraints achieve high refusal rates with minimal degradation in instruction following abilities.}}
    \label{fig:pareto_safety}
\end{figure}

\paragraph{Improvements in safety-helpfulness trade-offs.} 
As shown in Figure~\ref{fig:pareto_safety}, the model fine-tuned with pointwise safety constraints achieves 100\% refusal rate on  unsafe prompts while experiencing only a 5\% reduction in length controlled win rate (LC-WR) in AlpacaEval. In contrast, achieving the same refusal rate with average constraints or baseline methods reduces the LC-WR by 9\%, i.e. highly limiting the model's instruction following capabilities. This is explained by a considerably larger KL divergence with respect to the pre-trained model (right table). The resilient approach allows recovery of performance by relaxing constraints at the cost of lower refusal rates for unsafe prompts, while pareto dominating average constraints at higher refusal rates.

\paragraph{Length can be reduced without hurting relevance.}
This experiment is intended to illustrate the flexibility of constrained formulations, so there is no established benchmark for direct comparison. Nevertheless, the table in Figure~\ref{fig:pareto_reranker} shows that our method prefers shorter passages than models trained only for relevance, without degrading relevance performance. 
\begin{figure}[b]
    \centering
    \begin{minipage}[c]{0.52\linewidth}
        \centering
        \includegraphics[width=0.95\linewidth]{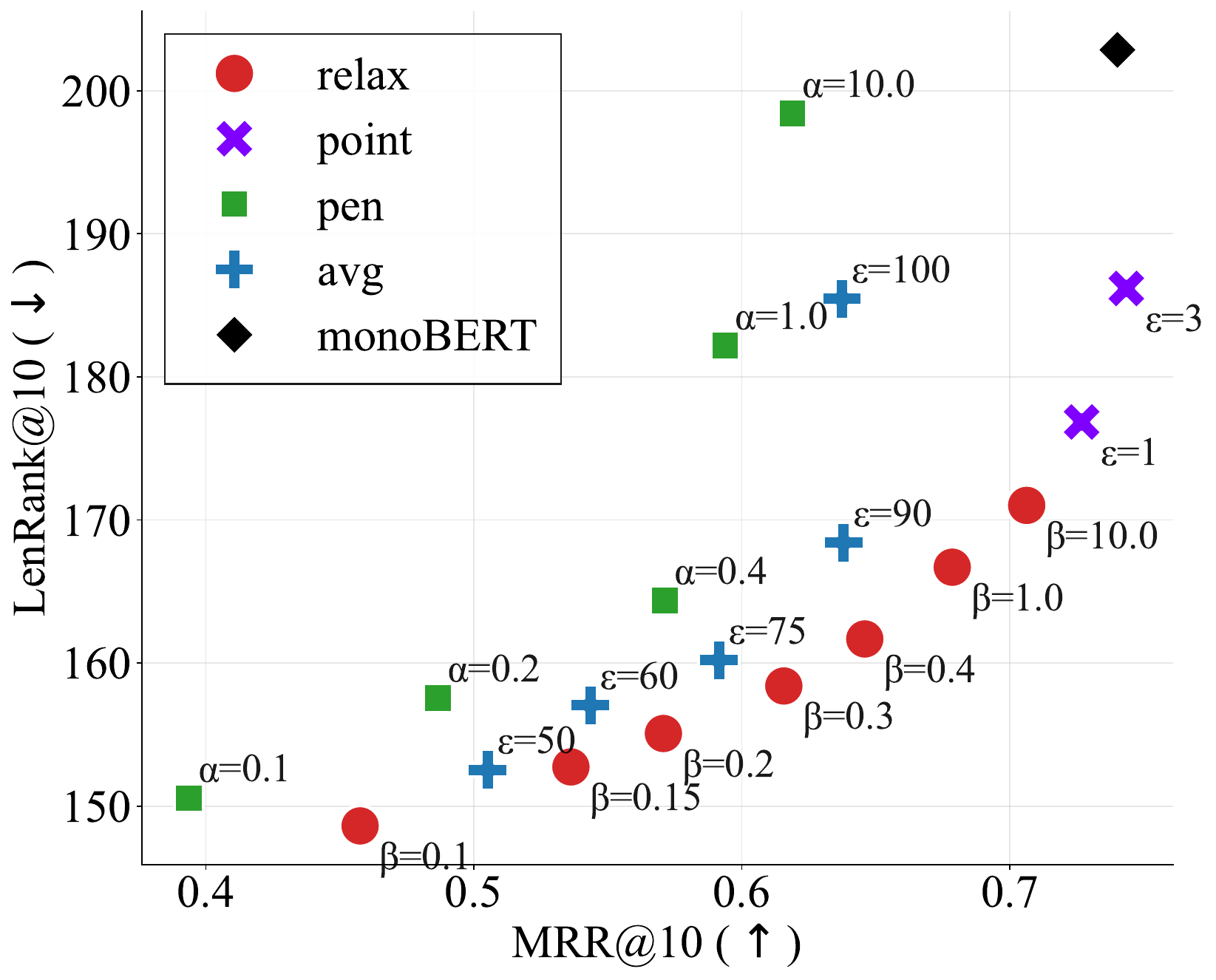}
        \end{minipage}
    \hfill
    \begin{minipage}[c]{0.46\linewidth}
        \centering
 \begin{tabular}{lcc}
    \hline
     & MRR & LenRank \\
    Method & $(\uparrow)$ & $(\downarrow)$ \\
    \hline
    MiniLM & 0.73 & 202 \\
    BGE-large & 0.73 & 205 \\
    monoBERT  & \textbf{0.74} \terr{1e-3} & 203 \terr{1} \\
    point ($\epsilon=3$) & 0.73 \terr{1e-2} & 183 \terr{2} \\
    point ($\epsilon=1$) & 0.72 \terr{1e-2} & 176 \terr{1} \\
    relax ($\beta=10$) & 0.71 \terr{5e-5} & \textbf{170.8} \terr{0.2} \\
    \hline
\end{tabular}
    \end{minipage}
    \caption{MRR@10 for ranking the positive passage versus LenRank@10 metric. The base model is ModernBERT-base~\cite{mordern-bert}, fine-tuned on a subset of MS MARCO~\cite{bajaj2016ms}. \textbf{Left}: Relevance--length tradeoff achieved by the different methods. Table~\ref{tab:reranking_all_results} includes results for a larger sweep of hyperparameters and includes the Hit@3 and AvgLen@3 metrics. \textit{Relaxed pointwise constraints define the Pareto frontier.} \textbf{Right}: Comparison against baselines trained only for relevance. \textit{Length can be reduced without significant degradation in relevance.}}
    \label{fig:pareto_reranker}
\end{figure}

\paragraph{Pointwise constraints pareto dominate the relevance--length tradeoff.}
The reduction in length without loss of relevance is only achieved when relevance is enforced through pointwise constraints. As shown in the image of Figure~\ref{fig:pareto_reranker}, alternative formulations incur a relevance drop as length decreases, whereas pointwise constraints can preserve relevance while shortening the retrieved context. Moreover, when some relevance degradation is acceptable in exchange for shorter outputs, pointwise constraints remain the strongest method overall, as across the full range of constraint relaxations and penalty values, the solutions obtained with pointwise constraints and resilience form the Pareto frontier.
%

%
%
%

%

\bibliography{refs}
\bibliographystyle{unsrtnat}

\appendix
\section*{Appendix}
The appendix is organized as follows:
\begin{itemize}
    \item Additional related work (Appendix~\ref{app:related}): a brief review of related approaches in constrained and robust learning,  and LM finetuning for the preference and safety tasks.
    \item Theoretical Results (Appendix~\ref{app:Theo}): Results supporting the claims made in the main text and their proofs.
    \item Algorithm (Appendix~\ref{app:algorithm}): Algorithm implementation and dual updates.
    \item Experimental Details (Appendix~\ref{app:appendix-exp-details}): Description of the different tasks, baselines, datasets and relevant hyperparameters used in our experiments.
    \item Experimental Results (Appendix~\ref{app:results}): Additional figures and tables, including negative results.    \item Ablations (Appendix~\ref{apx:ablations}): Ablations on standard versus augmented duality, sensitvity to the choiche of parameter $\alpha$ and computational costs.
\end{itemize}

\section{Additional related work}\label{app:related}
We compare first our approach to related learning paradigms that either use per sample constraints or seek to control the distribution of per sample losses. We then provide a brief review of related work for each of the tasks, focusing on their relation to our framework.

\paragraph{Pointwise Constraints in learning}
In statistical learning, per sample constraints were used as early as in Support Vector Machines~\cite{Vapnik1963PatternRU, SVMs} which require  classifying all points correctly while maximizing the classification margin. Soft margin SVMs~\cite{cortes1995support} can thus be interpreted as a particular case of our resilient framework with a linear relaxation cost $c(\bbu) = \beta\sum\bbu_i$. Motivated by soft margin SVMs and control applications, theoretical work in non-convex scenario optimization with relaxations~\cite{campirelaxedscenario, garatti2025nonconvscenario} provides statistical guarantees in terms of the probability of non-zero relaxations. Similarly, \cite{ramirez2025feasiblelearning} train neural networks to optimize \emph{a single loss} by solving a feasibility problem that requires every sample’s loss to be below a prescribed threshold. This approach thereby lacks an objective that establishes a preference among feasible solutions. Lastly, pointwise constraints over trajectories have also been studied in Reinforcement Learning settings~\cite{castellano22a-rl, anytime-rl, sootla2022saute}.

\paragraph{Robust Learning} Several approaches including exponential tilting~\cite{exp-tilting}, some forms of distributionally robust optimization and conditional value at risk minimization~\cite{dro-cvar, robey2022probabilistically, xu2025robust} have been proposed to control properties of the loss distribution beyond the expected risk such as tail behavior or worst-case performance under distribution shifts. 
While effective at shaping specific aspects of the loss distribution, these methods typically optimize a \emph{single} scalar objective and do not impose constraints.

\paragraph{Preference Optimization}
  Reinforcement Learning based approaches to learning from human feedback (RLHF) used binary preferences~\cite{stiennon2020learning} or rankings~\cite{ouyang2022training} to train a reward model, and then used this reward model to finetune the LM using Reinforcement Learning. These were followed by an extensive body of work on the design supervised preference losses that forego explicit reward models, using the preference labels directly to train the model. Most supervised preference losses optimize a transformation of the probability of LM outputs, as detailed in Table~\ref{app:tab-preferences-transformations}. Nonetheless, if we define the margin as
\begin{align*}
\text{margin}(\pi,\bx, \by_w, \by_l)  = \phi(\pi(\by_w|\bx), \by_w) - \phi(\pi(\by_l|\bx), \by_l ), 
\end{align*}
most methods can be interpreted as minimizing a penalty function on expectation, namely, 
\begin{align*}
    \min_{\theta} \quad & \mathbb{E}_{\bx, \by_w,\by_l \sim \mathcal{D}} \left[\text{penalty}(\text{-margin}(\pi,\bx, \by_w, \by_l))\right].
\end{align*}
Common penalty functions include the softplus (e.g., DPO~\cite{rafailov2024direct}, SimPO~\cite{meng2024simpo}) and ReLU (e.g., RRHF~\cite{yuan2023rrhf}, RePO~\cite{wu2025repo}). For explicit definitions of the losses used in our experiments see Table~\ref{tab:preference_losses}. 

A limitation of margin based losses is that they can be minimized while reducing the likelihood of both responses, whereas the likelihood of the preferred response should increase. Some works like SLiC-HF~\cite{Zhao2023SLiCHFSL} and CPO~\cite{xu2024cpo} have added a fixed penalty to increase the probability, explicitly,
\begin{align*}
    \min_{\theta} \quad & \mathbb{E}_{\bx, \by_w,\by_l \sim \mathcal{D}} \left[\text{margin}(\pi,\bx, \by_w, \by_l)-\lambda_0  \log \pi(\by_w|\bx)\right].
\end{align*}
\begin{table}[t]
    \caption{Common transformations applied to LM predicted probabilities in  preference optimization losses.}
    \label{app:tab-preferences-transformations}
    \centering
    \begin{tabular}{c|c}
    \hline
       $\phi(\pi(\by|\bx), \by)$ & Examples \\\hline      
         $\log\pi(\by|\bx)$    &  SLiC~\cite{Zhao2023SLiCHFSL}, CPO~\cite{xu2024cpo} \\
         $\frac{1}{|\by|}\log\pi(\by|\bx)$ & SimPO~\cite{meng2024simpo}, RRHF~\cite{yuan2023rrhf},  RePO~\cite{wu2025repo} \\
         $\log\frac{\pi(\by|\bx)}{\pi_{\text{ref}}(\by|\bx)}$& DPO~\cite{rafailov2024direct}, IPO~\cite{azar2024general}\\
         \hline
    \end{tabular}
\end{table}

In the same vein, Cal-DPO~\cite{xiao2024cal} adds a quadratic penalty for both prefered and disprefered responses to the DPO objective. In our framework, Cal-DPO can be interpreted as a penalty method that tackles the following relaxed pointwise constrained problem:
\begin{align*}
\min_{\theta, u} \quad & \mathbb{E}_{\bx, \by_w,\by_l \sim \mathcal{D}} \left[ -\log \sigma\left(\log \frac{\pi_\theta\left(\by_w \mid \bx\right)}{\pi_{\mathrm{ref}}\left(\by_w \mid \bx\right)}-\log \frac{\pi_\theta\left(\by_l \mid \bx\right)}{\pi_{\mathrm{ref}}\left(\by_l \mid \bx\right)}\right) + u^2(\bx,\by_w) + u^2(\bx,\by_l) \right]  \\
\text{s.t.} \quad & \log \pi_{\theta}(\by_w|\bx) = \log \pi_{\text{ref}} + \frac{1}{2\beta} +u(\bx,\by_w), \quad \mathcal{D}-a.e. \\
& \log \pi_{\theta}(\by_l|\bx) = \log \pi_{\text{ref}}(\by_l|\bx) -\frac{1}{2\beta} + u(\bx,\by_l), \quad \mathcal{D}-a.e, 
\end{align*}
The main difference with respect to our formulation is that we use the divergence to the pre-trained model as an objective, rather than the DPO loss. Moreover,  we tackle the relaxed primal through the augmented lagrangian dual, with per-sample multipliers, as explained in Sections~\ref{sec:resilience} and ~\ref{sec:algo}, rather than with a penalty method.

\paragraph{Safety} Reward models trained to align LMs with preference-based human feedback initially targeted both \emph{helpful} and \emph{harmless} completions~\cite{bai2022training}. 
Subsequently,~\cite{dai2024safe} trained separate reward and cost models for \emph{helpfulness} and \emph{harmlesness}, and proposed to maximize the helpfulness reward $R$ under an average constraint on the safety cost $C$, explicitly,
\begin{align}
\underset{\theta \in \Theta}{\operatorname{maximize}} \; \mathbb{E}_{\boldsymbol{x} \sim \mathcal{D}}\mathbb{E}_{\boldsymbol{y} \sim \pi_\theta}[R(\boldsymbol{x}, \boldsymbol{y})] \text{ subject to } \mathbb{E}_{\boldsymbol{x} \sim \mathcal{D}}\left[\mathbb{E}_{\boldsymbol{y} \sim \pi_\theta}\left[C(\boldsymbol{x}, \boldsymbol{y})\right]\right] \leq b
\end{align}
This formulation has been tackled by many works (see e.g. ~\cite{ji2025pku2, dai2024safe, wachi2024stepwise, huang2024one}) and extended to constraints over multiple costs~\cite{botong2025alignment} or unsafe content categories~\cite{yang2026cat-rl-multiple}.
Similar to other RLHF alignment tasks, recent work has demonstrated that standard supervised losses such as DPO~\cite{kim2025safedpo} and IPO~\cite{zhang2024bi-ipo}  --with minimal modifications to handle both safety and helpfulness labels-- can match or often exceed the performance of RL based methods.

Despite the rich body of work in safety tuning, the literature has focused primarily in using preference pairs to prevent performance degradation in helpfulness or instruction following capability (e.g. in ~\cite{dai2024safe, wachi2024stepwise, huang2024one, botong2025alignment, zhang2024bi-ipo, ji2025pku2, kim2025safedpo, qi2025midpo-MOE, niu2025mitigating-nullspace, yang2026cat-rl-multiple}). Moreover, helpfulness and harmlessness are evaluated for the same set of prompts. In many settings, the desirable response to unsafe queries varies depending on the type of query and user intent, with sometimes short, hard refusals being  appropriate~\cite{mu2024rule}. 

In contrast, we follow~\cite{bianchi2023safetyllama} and use an instruction following dataset commonly used for SFT to promote helpfulness-- i.e., without negative pairs -- and unsafe prompts paired with canonical refusals. Therefore our approach does not need paired labeled data, and relies only on increasing the probability of the desired refusal behaviour for unsafe prompts.

Lastly, other complementary approaches to safety alignment  include different parameterizations, such as LoRA based mixtures of experts~\cite{qi2025midpo-MOE}, constraints on parameter updates~\cite{niu2025mitigating-nullspace}, and multi-objective preference tuning~\cite{guo2024controllable-multiobj, yang2024metaaligner-multiobj}.

\paragraph{Pointwise Relaxations in other learning tasks.} The relaxed formulation in Section~\ref{sec:resilience} generalizes a classical idea in supervised learning: introducing pointwise slack variables to relax hard constraints while penalizing the amount of relaxation. A canonical example is the soft-margin SVM, where linear relaxation costs \(c(\bbu)\) are used to allow margin violations,
\begin{align}\label{eqn_svm}
	\min_{\bbtheta \in \Theta,\, \bbu \in \reals^N}. &\frac{1}{2} \norm{\bbtheta}^2 + \gamma \sum_{i = 1}^N u_i
	\\
	\text{s. to}& \;\;1-y_i \bbtheta^T \bbx_i \leq u_i
		\text{,} \quad i = 1,\dots,N
		\text{.}
\end{align}
We present the linear binary classification case for simplicity, although this framework has been extended to nonlinear settings through the so-called ``kernel trick'' and to regression tasks. The nominal pointwise constraint in~\eqref{eqn_svm}, \(\ell(\bx,\by,\theta)=1-y\bbtheta^\top\bbx\leq 0\), enforces correct classification with margin. In the hard-constrained case, minimizing \(\norm{\bbtheta}^2\) is equivalent to maximizing the classification margin. More generally, \citet{campirelaxedscenario} consider relaxed problems with general objectives and pointwise constraint functions, and provide bounds on the constraint violations of their solutions.

\section{Theoretical Results}\label{app:Theo}

\subsection{Average constraint formulation}
\label{sec:appx-avg-form}

We provide justification for the claim that the average constrained problem results in a scalar dual variable. The average constrained formulations replace the infinite number of constraints for a single one corresponding to the expected value taken over the same distribution, namely

\begin{align}\label{eqn:avg-formulation}
    \min_{\theta \in \Theta}. &
         \; \EE_{(\bx,\by)\sim\mathcal{D}}
        \Big[\ell_0(\bbx, \bby, \pi_\theta)\Big]   \\[0.2cm]
    \text{s. to} 
    & \; \EE_{(\bx,\by)\sim\mathcal{D}}\Big[\ell(\bx, \by, \pi_\theta)\Big] \leq \epsilon.
\end{align}

The Lagrangian associated to problem~\eqref{eqn:avg-formulation} is defined as
\begin{align}\label{eq:apx-lagrangian}
    \mathcal{L}(\pi_\theta,\blambda)
    \;:=\;&
    \mathbb{E}_{(\bx,\by)\sim\mathcal{D}}\;\![\ell_0\big(\bbx, \bby, \pi_\theta\big)\;] + \lambda\mathbb{E}_{(\bx,\by)\sim\mathcal{D}}\;\![\big(\ell(\bx,\bby,\pi_\theta)-\epsilon\big)] \\
    \;=\;& \mathbb{E}_{(\bx,\by)\sim\mathcal{D}}\;\![\ell_0\big(\bbx, \bby, \pi_\theta\big)\; + \lambda\big(\ell(\bx,\bby,\pi_\theta)-\epsilon\big)]
\end{align}

Thus, the dual problem is defined as:

\begin{equation}
    \max_{\lambda\in\reals_+} \min_{\theta \in \Theta} \mathbb{E}_{(\bx,\by)\sim\mathcal{D}}\;\![\ell_0\big(\bbx, \bby, \pi_\theta\big)\; + \lambda\big(\ell(\bx,\bby,\pi_\theta)-\epsilon\big)].
\end{equation}
\subsection{Dual Regularization and Optimal Relaxation}
\label{apx:dual_reg_apx}

We characterize the optimal relaxation $\bbu^\star$ as a balance between the improvement obtained by relaxing the constraints and the user-defined cost of doing so. Let the perturbation function associated with problem~\eqref{eqn:s-csft} be
\begin{align}\label{eqn:p-csft}
P(\bbu) =
\min_{\theta\,\in\,\Theta}
         \;  \;
          &\EE
        \Big[\ell_0(\bx, \bby, \pi_\theta)\Big]  
        \\[0.1cm]
        \subject \; \;
        & \ell(\bx, \by, \pi_\theta) 
        \leq
        \epsilon + \bbu(\bx,\by)  
        \quad \mathcal{D}\text{-a.e.} \nonumber
\end{align}
The function $P(\bbu)$ gives the best achievable objective value after relaxing the constraint by $\bbu$. Since larger relaxations enlarge the feasible set, $P$ is pointwise non-increasing. Thus, the relevant tradeoff is between the decrease in objective value induced by increasing $\bbu$ and the cost paid for that relaxation.

Let $\bar c:\reals_+\to\reals$ be a convex increasing pointwise cost, and define the functional
\begin{equation}
c(\bbu) := \EE\big[\bar c(\bbu(\bx,\by))\big].
\end{equation}

The relaxed problem can then be written as choosing the relaxation $\bbu$ that minimizes the total value
\begin{equation}
P(\bbu)+c(\bbu).
\end{equation}

This formulation makes explicit that the optimal relaxation balances marginal improvements in the objective against marginal increases in the relaxation cost.

To formalize this balance, we use subdifferentials. For a functional \(P:\mathcal L_p(\mathcal D)\to\overline{\reals}\), its subdifferential at \(u\in\mathcal L_p(\mathcal D)\) is
\begin{equation}
\partial P(u)
=
\left\{
z\in \mathcal L_q(\mathcal D):
P(v)-P(u)
\geq
\langle z,v-u\rangle
\ \text{for all } v\in\mathcal L_p(\mathcal D)
\right\},
\end{equation}

where
\begin{equation}
\langle z,v-u\rangle
=
\EE\big[z(\bx,\by)(v(\bx,\by)-u(\bx,\by))\big].
\end{equation}

Subgradients therefore play the role of generalized marginal sensitivities of the optimal value with respect to the relaxation.

\begin{proposition}[Optimal relaxation]\label{prop_res_primal_apx}
Consider the relaxed primal problem
\begin{align}\label{eqn_res_primal_apx}
(\pi_\theta^\star, \bbu^\star) \in 
\argmin_{\pi_\theta,\,\bbu}
         \;  \;
          &\EE
        \Big[\ell_0(\bx, \bby, \pi_\theta)\Big]
        + c(\bbu)
        \\[0.1cm]
        \subject \; \;
        & \ell(\bx, \by, \pi_\theta) 
        \leq
        \epsilon + \bbu(\bx,\by) 
        \quad \mathcal{D}\text{-a.e.}
\end{align}
Assume that the subdifferential sum rule holds at $\bbu^\star$. Then $\bbu^\star$ satisfies
\begin{equation}\label{eqn:res-eq}
0\in \partial P(\bbu^\star)+\partial c(\bbu^\star),
\end{equation}
and $\pi_\theta^\star$ is a solution of the perturbed problem \(P(\bbu^\star)\).
\end{proposition}

\begin{proof}
The joint minimization over $\pi_\theta$ and $\bbu$ in~\eqref{eqn_res_primal_apx} is equivalent to the nested problem
\begin{align}\label{nested}
    \min_{\bbu} \; P(\bbu)+c(\bbu).
\end{align}
Therefore, any optimal relaxation $\bbu^\star$ satisfies the first-order optimality condition
\begin{align}\label{eq:first-order-opt}
    0 \in \partial(P+c)(\bbu^\star).
\end{align}
By the subdifferential sum rule, under the stated regularity condition,
\[
\partial(P+c)(\bbu^\star)
=
\partial P(\bbu^\star)+\partial c(\bbu^\star).
\]
Combining this identity with~\eqref{eq:first-order-opt} yields
\[
0\in \partial P(\bbu^\star)+\partial c(\bbu^\star),
\]
which proves~\eqref{eqn:res-eq}. Since $\pi_\theta^\star$ minimizes the objective among models feasible under the relaxation $\bbu^\star$, it is a solution of the perturbed problem \(P(\bbu^\star)\).
\end{proof}

Proposition~\ref{prop_res_primal_apx} shows that the optimal relaxation is reached when the marginal improvement in the objective value is balanced by the marginal cost of relaxing the constraint. Equivalently,
\[
-\partial c(\bbu^\star) \cap \partial P(\bbu^\star) \neq \emptyset,
\]
so the cost of relaxation determines how much constraint difficulty is absorbed by $\bbu^\star$. For a quadratic cost, larger relaxations are assigned to regions where the perturbation function has larger sensitivity, i.e., where constraints are more costly to enforce.

\subsection{Proof of Proposition~\ref{prop:resilient_augmented_dual}}
\label{apx:proof-res}

\begin{proof}
Fix any $\theta\in\Theta$, $\blambda\in\mathbb{R}^N_+$, and $\alpha>0$. The Augmented Lagrangian of~\eqref{eqn_res_primal} for cost $c(\bu) = \beta E[\bu^2(\bx,\by)]$ can be written as

\begin{multline}\label{eq:res_lagrangian}
\mathcal{L}_{R}(\pi_\theta,\bu,\blambda, \alpha) =  \mathbb{E}\Big[\ell_0(\bx, \bby, \pi_\theta)  \; + \frac{\alpha}{4}\Big(2(\ell(\bx, \by, \pi_\theta)-\epsilon - \bu(\bx,\by)) +\frac{\lambda(\bx, \by)}{\alpha}\Big)_+^2 \\
- \frac{\lambda(\bx, \by)^2}{4\alpha} + \beta \bu^2(\bx,\by)\Big].
\end{multline}

Since the expression in
\eqref{eq:res_lagrangian} is separable in $(\bx,\by)$, we can minimize over each
$\bu(\bx,\by)$. For a fixed sample $(\bx,\by)$, let
\[
m(\bx,\by):= 2(\ell(\bx,\by,\pi_\theta)-\epsilon)+\frac{\lambda(\bx,\by)}{\alpha}.
\]
The terms of
$\mathcal{L}_R(\theta,\bu,\blambda,\alpha)$ that depend on $\bu(\bx,\by)$ are
\begin{align}
\phi(u(\bx,\by))
&:=\beta \bu^2(\bx,\by)+\frac{\alpha}{4}\max\{0,m(\bx,\by)-2\bu(\bx,\by)\}^2.
\label{eq:phi-def}
\end{align}

Consider two cases.

\begin{itemize}
    \item If $m(\bx,\by) \leq 0$:
\end{itemize}
Then $\max\{0,m(\bx,\by)-2\bu(\bx,\by)\}=0$ for all $u\ge 0$, and \eqref{eq:phi-def} reduces to
$\phi(u(\bx,\by))=\beta \bu^2(\bx,\by)$, which is minimized at $0$, giving
\begin{equation}
\min_{u\ge 0}\phi(u(\bx,\by))= 0
\label{eq:case1}
\end{equation}

\begin{itemize}
    \item If $m(\bx,\by)>0$:
\end{itemize}
For $0\le \bu(\bx,\by)<m(\bx,\by)/2$, $\max\{0,m(\bx,\by)-2\bu(\bx,\by)\}=m(\bx,\by)-2\bu(\bx,\by)$ and \eqref{eq:phi-def} is a strictly convex quadratic:
\begin{align}
\phi(u(\bx,\by))
&=\beta \bu^2(\bx,\by)+\frac{\alpha}{4}(m(\bx,\by)-2\bu(\bx,\by))^2 \\
&= (\alpha+\beta)\bu^2(\bx,\by)-\alpha m(\bx,\by) \bu(\bx,\by)+\frac{\alpha}{4}m(\bx,\by)^2
\end{align}
Taking derivatives w.r.t. $\bu(\bx,\by)$ we get the first order  condition $2(\alpha+\beta)\bu(\bx,\by) = \alpha m(\bx,\by)$, which can be rearranged as
\begin{align}~\label{eq_u_first_order}
\bu(\bx,\by) =\frac{\alpha}{2(\alpha+\beta)}m(\bx,\by).
\end{align}
Since $\alpha, \beta>0$, any $u$ satisfying~\eqref{eq_u_first_order} lies in $(0,m(\bx,\by)/2)$. 
Substituting in $\phi$ yields
\begin{align}
\min_{u\ge 0}\phi(u(\bx,\by))
&= \frac{\alpha\beta}{4(\alpha+\beta)}m(\bx,\by)^2.
\label{eq:case2}
\end{align}

Combining \eqref{eq:case1} and \eqref{eq:case2}, we obtain
\begin{align}\label{apx:point-res-proof}
\min_{\bu(\bx,\by) \ge 0} & \phi (\bu(\bx,\by)) + \ell_0(\bx, \bby, \pi_\theta) - \frac{\lambda(\bx, \by)^2}{4\alpha} 
\\ = \min_{\bu(\bx,\by) \ge 0} &  \frac{\alpha\beta}{4(\alpha+\beta)}\Big(m(\bx,\by)\Big)_+^2  + \ell_0(\bx, \bby, \pi_\theta) - \frac{\lambda(\bx, \by)^2}{4\alpha}.
\end{align}
Taking the expected value of~\eqref{apx:point-res-proof} yields
\begin{align}\label{apx:lag_minimized_res_proof}
\min_{\bu \in L^\infty}\ \mathcal{L}_R(\pi_\theta,\bu,\blambda,\alpha)
=
\mathbb{E}\Bigg[
\ell_0(\bx,\bby,\pi_\theta)
& +
\frac{\alpha\beta}{4(\alpha+\beta)}
\left(
2(\ell(\bx,\by,\pi_\theta)-\epsilon)
+
\frac{\lambda(\bx,\by)}{\alpha}
\right)_+^2 \\
&-
\frac{\lambda(\bx,\by)^2}{4\alpha}
\Bigg].
\end{align}
To show that the expression in \ref{apx:lag_minimized_res_proof} is equivalent to \eqref{eq:res_aug_lagrangian_dual}  let 
\[
\tilde \alpha := \frac{\alpha\beta}{\alpha+\beta},
\qquad
\tilde \lambda(\bx,\by)
:=
\frac{\beta}{\alpha+\beta}\lambda(\bx,\by).
\]
Noticing that
\[
\frac{\tilde \lambda(\bx,\by)}{\tilde \alpha}
=
\frac{\lambda(\bx,\by)}{\alpha},
\]
we get the equivalence
\begin{align}
\frac{\alpha\beta}{4(\alpha+\beta)}
\left(
2(\ell(\bx,\by,\pi_\theta)-\epsilon)
+
\frac{\lambda(\bx,\by)}{\alpha}
\right)_+^2
=
\frac{\tilde \alpha}{4}
\left(
2(\ell(\bx,\by,\pi_\theta)-\epsilon)
+
\frac{\tilde \lambda(\bx,\by)}{\tilde \alpha}
\right)_+^2 .
\end{align}
Moreover, as
\begin{align}
-\frac{\lambda(\bx,\by)^2}{4\alpha}
&=
-\frac{\tilde\lambda(\bx,\by)^2}{4\tilde\alpha}
-
\frac{\lambda(\bx,\by)^2}{4(\alpha+\beta)}.
\end{align}
the expression in \eqref{apx:lag_minimized_res_proof} is equivalent to
\begin{multline}
\min_{\bu \in L^\infty}\ \mathcal{L}_R(\pi_\theta,\bu,\blambda,\alpha)
=
\mathbb{E}\Bigg[
\ell_0(\bx,\bby,\pi_\theta)
+
\frac{\tilde \alpha}{4}
\left(
2(\ell(\bx,\by,\pi_\theta)-\epsilon)
+
\frac{\tilde \lambda(\bx,\by)}{\tilde \alpha}
\right)_+^2
-
\frac{\tilde\lambda(\bx,\by)^2}{4\tilde\alpha}
\Bigg]
\\
\qquad
-
\frac{1}{4(\alpha+\beta)}
\mathbb{E}\left[\lambda(\bx,\by)^2\right].
\end{multline}
The first expectation is precisely the augmented Lagrangian associated with
the unrelaxed problem, evaluated at
\[
\tilde \lambda = \frac{\beta}{\alpha+\beta}\lambda,
\qquad
\tilde \alpha = \frac{\alpha\beta}{\alpha+\beta}.
\]
Thus,
\begin{align}
\min_{\bu \in L^\infty}\ \mathcal{L}_R(\pi_\theta,\bu,\blambda,\alpha)
=
\mathcal{L}_A
\left(
\pi_\theta,
\frac{\beta}{\alpha+\beta}\blambda,
\frac{\alpha\beta}{\alpha+\beta}
\right)
-
\frac{1}{4(\alpha+\beta)}
\|\blambda\|_2^2 .
\end{align}
Finally, minimizing over $\theta\in\Theta$ gives
\begin{align}
g_R(\blambda,\alpha)
&:=
\min_{\theta\in\Theta,\ \bu\in L^\infty}
\mathcal{L}_R(\pi_\theta,\bu,\blambda,\alpha)
\\
&=
g_A\left(
\frac{\beta}{\alpha+\beta}\blambda,
\frac{\alpha\beta}{\alpha+\beta}
\right)
-
\frac{1}{4(\alpha+\beta)}
\|\blambda\|_2^2 .
\end{align}
\end{proof}

\subsection{Strong Duality of the Empirical Relaxed Lagrangian Dual}

In this section we relate the optimal value of the empirical resilient primal, explicitly
\begin{align}\label{eqn_res_primal_empirical}
\hat{P} = 
&\min_{\theta \in \Theta \,,\bbu\in\reals^N_+}
         \; 
          \frac{1}{N}\sum_{i=1}^N
       \left[\ell_0(\bx_j,\bby, \pi_\theta )  + c(\bbu_i)\right]
        \\[0.1cm]
        & \text{s. to } \; \;
     \ell(\bx_i, \by_i, \pi_\theta) \, 
        \;\leq\;
        \epsilon + \bbu_i \\
        & \quad \quad   \forall \;\; i=1,\ldots, N,
\end{align}
to the optimal value of the augmented lagrangian dual:
\begin{align}\label{eq:aug_lagrangian_dual_og_emp}
\hat{D} = \max_{\blambda \in\reals^N_+, \alpha \in \mathbb{R}_+} \min_{\theta \in \Theta \,,\bbu\in\reals^N_+} \hat{\mathcal{L}}(\theta, \bbu,\blambda,\alpha),
\end{align}
where 
\begin{align}\label{eq:aug_lagrangian_og_emp}
\hat{\mathcal{L}}(\theta, \bbu,\blambda,\alpha) =  \sum_{i=1}^N & \frac{1}{N}\left[\ell_0(\bx_i, \bby_i, \pi_\theta) + c(\bbu_i)\right] \\
        & \; + \alpha \Psi\left(\ell(\bx_i, \by_i, \pi_\theta)-\epsilon-\bbu_i, \frac{\lambda_i}{\alpha}\right),
\end{align}
denotes the empirical augmented lagrangian.

Analogous to the non-augmented dual problem, in general~\eqref{eq:aug_lagrangian_dual_og_emp} lower bounds the primal, i.e., $P^\star\geq D^\star$ (weak duality). However, the Augmented Lagrangian guarantees strong duality ($P^\star=D^\star$) even in non-convex settings under mild conditions as stated next. 

\begin{assumption}[Bounded Objective]
\label{ass:l0_bounded}
The function $\ell_0$ is bounded below.
\end{assumption}

\begin{assumption}[Compact Domain]
\label{ass:theta_closed}
The set $\Theta$ is compact.
\end{assumption}

\begin{assumption}[Continuity]
\label{ass:l_continuity}
The mapping $\theta \mapsto \pi_{\theta}$ is continuous, and the sample-wise losses 
${\ell}_i$ are continuous and uniformly bounded in~$\pi_\theta(\bby|\bx)$.
\end{assumption}

Assumption~\ref{ass:l0_bounded} holds whenever the sample-wise losses $\bar{\ell}_0$ are bounded below, 
which holds for several commonly used losses such as MSE or cross-entropy. 
Assumptions~\ref{ass:theta_closed} and~\ref{ass:l_continuity} impose only mild regularity on the parameterization mapping and the constraint functions. 
Note that unlike the hard constrained case which requires an additional assumption on feasible set compactness (see e.g.~\cite{boero2025cole}[Assumption 2.4]) that plays the role of
Slater’s constraint qualification in non-convex settings, the resilient relaxation does not, as formalized next.

\begin{proposition}[Strong Augmented Duality]~\label{prop:strong-dual-aug-res}
    If Assumptions~\ref{ass:l0_bounded}--\ref{ass:l_continuity} hold, then strong duality holds
    \begin{align}
        \hat{D} = \hat{P}
    \end{align}

\end{proposition}
\begin{proof}
From~\citep[Theorem 4]{RockaDual}, strong duality $\hat{D} = \hat{P}$ holds if the problem satisfies two conditions: the \emph{Growth Condition} and \emph{Stability of Degree 0}. 
We verify that these two conditions hold under the assumptions~\ref{ass:l0_bounded}--\ref{ass:l_continuity}.

First, we define the perturbation function $\nu: \mathbb{R}^N \to \mathbb{R}$, which represents the optimal value of the primal problem when the constraints are perturbed by a vector $\mathbf{z} \in \mathbb{R}^N$:
\begin{align}
    \nu(\mathbf{z}) = 
    &\min_{\theta \in \Theta \,,\bbu\in\reals^N_+}
          \; 
          \frac{1}{N}\sum_{i=1}^N
       \left[\ell_0(\bx_i,\bby_i, \pi_\theta )  + c(\bbu_i)\right]
        \\[0.1cm]
        & \text{s.t. } \; \;
     \ell(\bx_i, \by_i, \pi_\theta) \, 
        \;\leq\;
        \epsilon + \bbu_i + z_i, \quad \forall\; i=1,\ldots, N.
\end{align}

\paragraph{Growth Condition.}
The problem satisfies the growth condition if the augmented Lagrangian (with dual variables set to zero) is bounded below. Consider the Lagrangian at $\blambda = \mathbf{0}$:
\begin{align}
    \hat{\mathcal{L}}(\theta, \bbu, \mathbf{0}, \alpha) &= \sum_{i=1}^N \left[ \frac{1}{N}\ell_0(\bx_i, \bby_i, \pi_\theta) + \frac{1}{N}c(\bbu_i) + \alpha \Psi\left(\ell(\bx_i, \by_i, \pi_\theta)-\epsilon-\bbu_i, 0\right) \right] \\
    &= \sum_{i=1}^N \left[ \frac{1}{N}\ell_0(\bx_i, \bby_i, \pi_\theta) + \frac{1}{N}c(\bbu_i) + \frac{\alpha}{4} \operatorname{max}^2 \{0, 2(\ell(\bx_i, \by_i, \pi_\theta)-\epsilon-\bbu_i)\} \right].
\end{align}
By Assumption~\ref{ass:l0_bounded}, $\ell_0$ is bounded below. Since the relaxation cost $c(\bbu_i)$ (e.g., quadratic) and the penalty term (squared max function) are non-negative, the overall expression is bounded below for any $\alpha > 0$. Thus, the growth condition holds.

\paragraph{Stability of Degree 0.}

We must show that the perturbation function $\nu(\mathbf{z})$ is lower semicontinuous at $\mathbf{z}=\mathbf{0}$. We show a stronger property: that $\nu(\mathbf{z})$ is \emph{continuous} near $\mathbf{0}$ due to the presence of the relaxation variables $\bbu$.

By Assumption~\ref{ass:l_continuity}, the mapping $\theta \mapsto \pi_\theta$ is continuous, and the losses $\ell$, $\ell_0$ and the cost $c(\bbu)$ are continuous. Consequently, the objective function $\mathcal{J}(\theta, \bbu) = \frac{1}{N}\sum \ell_0(\cdot) + c(\bbu)$ is continuous on $\Theta \times \mathbb{R}^N_+$.

Then we can write the perturbation function as
\begin{align}
    \nu(\mathbf{z}) = \min_{\theta, \bbu} \left\{ \mathcal{J}(\theta, \bbu) \mid \ell_i(\theta) \le \epsilon + \bbu_i + z_i, \; \bbu \ge 0 \right\}.
\end{align}

Unlike standard constrained optimization where feasible sets may become empty under perturbation, the resilient formulation guarantees feasibility for any $\mathbf{z}$ and any $\theta$. For any $\theta$, the constraints can be satisfied by choosing $\bbu_i \ge \max(0, \ell_i(\theta) - \epsilon - z_i)$. 

Let us define the \emph{effective cost} function $\tilde{J}(\theta, \mathbf{z})$ by analytically minimizing over $\bbu$:
\begin{align}
    \tilde{J}(\theta, \mathbf{z}) = \frac{1}{N}\sum_{i=1}^N \ell_0(\bx_i, \bby_i, \pi_\theta) +\frac{1}{N} \sum_{i=1}^N \min_{\substack{\bbu_i \ge 0 \\ \bbu_i \ge \ell_i(\theta) - \epsilon - z_i}} c(\bbu_i).
\end{align}
Since $c(\cdot)$ is a continuous, non-decreasing function (e.g., quadratic) on $\mathbb{R}_+$, the inner minimization is continuous with respect to the bound $\ell_i(\theta) - \epsilon - z_i$. Consequently, $\tilde{J}(\theta, \mathbf{z})$ is jointly continuous in $\theta$ and $\mathbf{z}$.

Then we can re-write the perturbation function as the unconstrained minimization of this effective cost:
\begin{align}
    \nu(\mathbf{z}) = \min_{\theta \in \Theta} \tilde{J}(\theta, \mathbf{z}).
\end{align}
Since $\nu(\mathbf{z})$ is the minimum of a continuous function over a compact domain, it is continuous for any $\mathbf{z}$. In particular, it is continuous at $\mathbf{z}=\mathbf{0}$, thus satisfying the stability condition.

\end{proof}

\section{Algorithm}\label{app:algorithm}

We provide additional details for the inexact empirical dual-ascent procedure used
to solve the empirical augmented and relaxed dual problems in
\eqref{eq:emp-dual-aug} and \eqref{eq:emp-dual-res}. Throughout the appendix, let
\[
v_i(\pi_\theta)
:=
\ell(\bx_i,\by_i,\pi_\theta)-\epsilon
\]
denote the constraint violation of sample $i$.

By Danskin's theorem~\cite[Prop.~B.22]{Bertsekas01031997}, if the primal minimizer is unique, the gradient of the
empirical dual function with respect to $\blambda$ is obtained by differentiating
the empirical Lagrangian with respect to $\blambda$ at the primal minimizer.
When the minimizer is not unique, the same expression gives a valid dual
supergradient. Therefore, if
\begin{equation}
\hat \pi_{\theta}^A(\blambda, \alpha)
\in
\argmin_{\theta\in \Theta}
\hat{\mathcal L}_A(\theta,\blambda,\alpha),
\end{equation}

then a coordinate supergradient of the augmented empirical dual is
\begin{align}\label{eq:dual-direction}
\frac{\partial g^A}{\partial \lambda_i}(\blambda,\alpha)
&:=
\begin{cases}
v_i(\hat \pi_{\theta}^A(\blambda,\alpha)),
&
v_i(\hat \pi_{\theta}^A(\blambda,\alpha))\ge -\frac{\lambda_i}{2\alpha},
\\[4pt]
-\frac{\lambda_i}{2\alpha},
&
v_i(\hat \pi_{\theta}^A(\blambda,\alpha))< -\frac{\lambda_i}{2\alpha}.
\end{cases}
\end{align}

For the relaxed empirical dual, the corresponding coordinate supergradient is

\begin{align}\label{eq:dual-direction-res}
\frac{\partial g^R}{\partial \lambda_i}(\blambda,\alpha)
=
\begin{cases}
\displaystyle
\frac{\beta}{\alpha+\beta}v_i(\pi_{\theta}^R(\blambda,\alpha))
-
\frac{\lambda_i}{2\beta(\alpha+\beta)},
&
v_i(\pi_{\theta}^R(\blambda,\alpha))\ge -\frac{\lambda_i}{2\alpha},
\\[8pt]
\displaystyle
-\frac{\beta\lambda_i}{2\alpha(\alpha+\beta)}
-
\frac{\lambda_i}{2\beta(\alpha+\beta)},
&
v_i(\pi_{\theta}^R(\blambda,\alpha))< -\frac{\lambda_i}{2\alpha}.
\end{cases}
\end{align}

The supergradient in \eqref{eq:dual-direction-res} include a shrinking term induced by the relaxation cost that was not present in \ref{eq:dual-direction}. It prevents
multipliers associated with persistently infeasible or outlying samples from
growing indefinitely.

Although the empirical dual problems are written as maximizations over both
$\blambda$ and $\alpha$, our implementation fixes $\alpha$ and performs ascent
only over $\blambda$. This corresponds to the shifted-penalty method in the
augmented-Lagrangian literature~\cite[Algorithm~7.3]{shifted}. The terminology
comes from the fact that the quadratic penalty is applied to the shifted
violation
\begin{equation}
\ell(\bx_i,\by_i,\pi_\theta)-\epsilon+\frac{\lambda_i}{2\alpha},
\end{equation}

rather than directly to the violation
$\ell(\bx_i,\by_i,\pi_\theta)-\epsilon$. Hence, for sample $i$, the penalty
becomes active when
\begin{equation}
\ell(\bx_i,\by_i,\pi_\theta)-\epsilon
\ge
-\frac{\lambda_i}{2\alpha},
\end{equation}
so the multiplier $\lambda_i$ shifts the point at which violations begin to be
penalized.

Fixing $\alpha$ is justified by a standard monotonicity property of the
augmented dual. For fixed $\blambda$, the augmented Lagrangian is nondecreasing
in $\alpha$, and therefore the augmented dual function is also nondecreasing in
$\alpha$. Since every dual value is upper bounded by the primal optimum, if
$(\blambda^\star,\alpha^\star)$ is an optimal dual pair, then for any
$\alpha_2\ge \alpha^\star$,

\begin{equation}
g_A(\blambda^\star,\alpha^\star)
\le
g_A(\blambda^\star,\alpha_2)
\end{equation}

which implies that $(\blambda^\star,\alpha_2)$ is also optimal. Thus, fixing
$\alpha$ is sufficient provided it is chosen above the problem-dependent
minium necessary value.

A common extension to the shift penalty method is the increased shifted-penalty method ~\cite[Algorithm~7.4]{shifted}, which alternates
several updates of $\blambda$ with occasional increases of the augmentation
parameter, for example $\alpha \leftarrow c\alpha$ with $c>1$. In our
experiments, however, we observed stable performance across a range of values of
$\alpha$. 
We therefore use the simpler fixed-$\alpha$
variant throughout.

In practice, computing the exact primal minimizer at every dual iteration is
prohibitively expensive. Instead, starting from the current model, we take a
small number of stochastic optimizer steps on the empirical Lagrangian. That is,
for a fixed $\blambda_t$, we compute an approximate minimizer $\hat \theta_t$
satisfying
\begin{align}
\hat{\mathcal L}_{A}(\hat\theta_t,\blambda_t,\alpha)
\le
\min_\theta
\hat{\mathcal L}_{A}(\theta,\blambda_t,\alpha)
+
\delta_t.
\end{align}
This produces a $\delta_t$-approximate dual supergradient. Indeed, for any
$\blambda'$, differentiating the Lagrangian at $\hat\theta_t$ gives a direction
$d_t$ satisfying
\[
\hat g(\blambda')
\le
\hat g(\blambda_t)
+
\langle d_t,\blambda'-\blambda_t\rangle
+
\delta_t.
\]

In particular, the convergence of the exact shifted-penalty method is recovered for the inexact approximation when the primal error
$\delta_t$ vanishes~\citep[Theorem~3.1]{boero2025cole}.

Finally, because evaluating the full empirical Lagrangian and all constraints at
every iteration is prohibitive for large datasets, both the primal and dual
updates are computed on minibatches. For the primal update, this replaces the
full gradient of $\hat{\mathcal L}_{A/R}$ with a stochastic gradient computed on
the current batch, as is standard in large-scale machine learning. For the dual
update, we update only the multipliers associated with samples in the current
batch, yielding a stochastic coordinate-ascent update on the empirical dual
variables. We choose S=1; at each step we update the dual variables of the current mini-batch and update the LM parameters given the stochastic gradient computed on that mini-batch. As the violations $v_i(\theta)$ are computed during
the same forward pass used to evaluate the primal objective, so the multiplier
update requires no additional forward passes.

\begin{algorithm}[h]
\begin{algorithmic}[1]
\State \textbf{Input:} init model $\pi_{0}$, init dual $\blambda_{0}$, dual lr $\eta_t$, primal lr $\gamma_t$,total iterations $T$, primal steps per iteration $S$
\For{$t = 0,\ldots,T-1$}
  \State Dual gradient ascent:
  \Statex \hspace{\algorithmicindent} $\blambda_{t+1}=\blambda_t+\eta_t\, \nabla_{\blambda}\hat g_{A/R}(\blambda)$.
  \State LLM optimization steps:
  \For{$s = 0,\ldots,S-1$}
    \State $\theta^{s+1}_t
    =
    \theta^s_t
    -
    \gamma_t
    \nabla_{\theta}
    \hat{\mathcal{L}}_{A/R}
    \big(\theta^s_t,\blambda_t,\alpha\big)$
  \EndFor
\EndFor
\State \textbf{Output:} $\theta_T, \blambda_T$.
\end{algorithmic}
\caption{Augmented Dual Ascent}
\label{alg:primal-dual}
\end{algorithm}
\section{Tasks details}
\label{app:appendix-exp-details}

We provide a detailed explanation for each of the tasks briefly presented in section \ref{sec:problem-form}, together with a description of baselines and hyperparameters.

\subsection{Function Calling Preferences}
\label{app:function_calling}
 Despite the increasing performance of SLMs in  function calling benchmarks such as  ToolBench \citep{toolllm} and the Berkeley Function Calling Leaderboard (BFCL) \citep{bfcl}, which primarily focus on the accuracy of tool selection and parameter filling, recent research highlights 
 common failure modes when the correct
tool is not provided or the user does not provide enough information to solve the task~\citep{toolbehonest, toolsandbox,when2call}. We use the \texttt{When2call} dataset~\cite{when2call} that aims to evaluate the decision-making capabilities of language models (LMs) when interacting with external tools. Specifically, the model must choose between four distinct behaviors for a given user query $x$:
\begin{enumerate}
    \item\textbf{ Direct answer:} A text response without using tools (which, in the context of this benchmark, constitutes a hallucination).
    \item \textbf{ Tool call:} A correctly formatted call to an available tool.
    \item \textbf{Follow-up question:} Requesting missing information required by the tool parameters.
    \item \textbf{ Unable to answer:} Correctly stating the request cannot be fullfilled with the available tools.
\end{enumerate}
\citep{when2call} construct a preference dataset including the correct option as the preferred response and one disprefered response randomly sampled from the incorrect options. We utilize this preference dataset to train the model, minimizing the expected KL divergence between the policy $\pi_{\theta}$ and a reference policy $\pi_{\text{ref}}$ while enforce pointwise constraints on the length-normalized log-probabilities of both the winning (preferred) option $y_w$ and the losing (dispreferred) option $y_l$. The resulting constrained optimization problem is defined as:
\begin{equation}
\begin{aligned}
\min_{\theta} \quad & \mathbb{E}_{x \sim \mathcal{D}} \left[ D_{\text{KL}}(\pi_{\theta}(\cdot|x) \| \pi_{\text{ref}}(\cdot|x)) \right] \\
\text{s.t.} \quad & \frac{1}{|y_w|} \log \pi_{\theta}(y_w|x) \geq \epsilon_{\text{win}}, \quad \mathcal{D}-a.e. \\
& \frac{1}{|y_l|} \log \pi_{\theta}(y_l|x) \leq \epsilon_{\text{lose}}, \quad \mathcal{D}-a.e.
\end{aligned}
\end{equation}
where $\epsilon_{\text{win}}$ and $\epsilon_{\text{lose}}$ represent the per-sample thresholds for the preferred and dispreferred behaviors, respectively. By enforcing these constraints at the sample level, we ensure that the model consistently reaches a minimum confidence for the correct behavior while remaining below a maximum tolerance for incorrect behaviors.

\paragraph{Baselines} We utilize as baselines the preference loses DPO~\cite{rafailov2024direct}, SimPO~\cite{meng2024simpo}, CalDPO~\cite{xiao2024cal} and CPO~\cite{xu2024cpo}, as detailed in Table~\ref{tab:preference_losses}.

\begin{table}[]
\caption{Baseline Preference Optimization Loss Functions used in experiments}
\label{tab:preference_losses}
\centering
\begin{tabular}{ll}
\toprule
\textbf{Method} & \textbf{Loss Formulation} \\ \midrule
DPO~\cite{rafailov2024direct} & $-\mathbb{E}_{(x, y_w, y_l) \sim \mathcal{D}} \left[ \log \sigma \left( \beta \log \frac{\pi_\theta(y_w|x)}{\pi_{\text{ref}}(y_w|x)} - \beta \log \frac{\pi_\theta(y_l|x)}{\pi_{\text{ref}}(y_l|x)} \right) \right]$ \\ \addlinespace
SimPO~\cite{meng2024simpo} & $-\mathbb{E}_{(x, y_w, y_l) \sim \mathcal{D}} \left[ \log \sigma \left( \frac{\beta}{|y_w|} \log \pi_\theta(y_w|x) - \frac{\beta}{|y_l|} \log \pi_\theta(y_l|x) - \gamma \right) \right]$ \\ \addlinespace
CalDPO~\cite{xiao2024cal} & $\mathcal{L}_{\text{DPO}} + \left(\log \frac{\pi_\theta\left(\mathbf{y}_w \mid \mathbf{x}\right)}{\pi_{\mathrm{ref}}\left(\mathbf{y}_w \mid \mathbf{x}\right)}-\frac{1}{2 \beta}\right)^2+\left(\log \frac{\pi_\theta\left(\mathbf{y}_l \mid \mathbf{x}\right)}{\pi_{\mathrm{ref}}\left(\mathbf{y}_l \mid \mathbf{x}\right)}+\frac{1}{2 \beta}\right)^2$ \\ \addlinespace
CPO~\cite{xu2024cpo} & $ -\mathbb{E}_{(x, y_w, y_l) \sim \mathcal{D}} \left[ \log \sigma \left( \log \pi_\theta(y_w|x) - \log \pi_\theta(y_l|x) \right)-\lambda_0 \log \pi_\theta(y_w|x) \right]$ \\\addlinespace \bottomrule
\end{tabular}
\end{table}

\paragraph{Dataset} The When2Call dataset \cite{when2call}, which consists of queries derived from the Berkeley Function Calling Leaderboard (BFCL)~\cite{bfcl} and APIGen~\cite{apigen}. The dataset was constructed by taking successful tool-calling examples and synthetically generating variants where information is missing, requiring a \textit{follow-up question}, or the query is outside the scope of provided tools, requiring an \textit{unable to answer} response. The training set includes 4,500 tool-calling examples and 3,000 examples each for the follow-up and ``unable to answer'' categories. Unlike \cite{when2call}, we hold out a random 5\% subset of the training split as a validation set.

\paragraph{Evaluation} We follow the evaluation protocol from~\cite{when2call}, which uses an offline multiple-choice format. For each test prompt, the model's chosen behavior is determined by the highest probability among four possible candidate answers. We report the following metrics:
\begin{itemize}
    \item \textbf{F1 Score:} Calculated per category.
    \item \textbf{Length Normalised  Accuracy:} The mean accuracy weighted by response length.
    \item \textbf{Hallucination Rate:} The frequency with which the model incorrectly chooses a ``Direct answer''.
\end{itemize}
All of the test scores use the full test split from the original dataset.

\paragraph{Hyperparameters} We use LoRA~\cite{hu2022lora} and AdamW~\cite{adamw} optimizer, together with other shared hyperparameters, as detailed in Table~\ref{tab:when2shared_hyperparams}, across all methods. We conducted a grid search for the Llama-3.2-1B-Instruct model to choose the hyperparameters for each method using the search spaces specified in Table~\ref{tab:when2method_hyperparams}. The test performance reported in our main results corresponds to the best-performing parameters as evaluated the held-out set. For the xLAM models, we utilized the same hyperparameter configurations identified for the Llama-3.2-1B model without further hyperparameter search. Although further hyperparameter tuning for xLAM might yield improved performance, our goal was not to achieve state-of-the-art results for a single architecture, but rather to demonstrate that our findings generalize across different models.

\begin{table}[h]
\centering
\begin{tabular}{l|l}
\toprule
\textbf{Category} & \textbf{Parameters} \\
\hline
Training &
$\begin{aligned}
\text{batch size} &= 32 \\
\text{learning rate} &= 5 \times 10^{-5} \\
\text{lr scheduler} &= \text{cosine} \\
\text{warmup ratio} &= 0.0 \\
\text{weight decay} &= 0.0 \\
\text{precision} &= \text{bf16} \\
\text{max prompt length} &= 1024 \\
\text{max length} &= 2048
\end{aligned}$ \\
\hline
LoRA &
$\begin{aligned}
r &= 8 \\
\text{lora\_alpha} &= 16 \\
\text{lora\_dropout} &= 0.05
\end{aligned}$ \\
\bottomrule
\end{tabular}
\caption{Shared training and model parameters for the function calling task.}\label{tab:when2shared_hyperparams}
\end{table}

\begin{table}[h]
\centering
\begin{tabular}{l|l}
\toprule
\textbf{Method} & \textbf{Hyperparameter Grid} \\
\hline
SimPO &
$\begin{aligned}
\gamma &\in [\underline{0}, 1, 2] \\
\beta &\in [\underline{2.0}, 2.5] \\
\text{epochs} &\in \{1,5,\underline{10}\}
\end{aligned}$ \\
\hline
DPO &
$\begin{aligned}
\beta &\in \{0.05,\underline{0.1},0.5\} \\
\text{epochs} &\in \{1,5,\underline{10}\}
\end{aligned}$ \\
\hline
Cal-DPO &
$\begin{aligned}
\beta &\in \{\underline{0.1},0.5,1.0\} \\
\text{epochs} & =10
\end{aligned}$ \\
\hline
CPO &
$\begin{aligned}
\lambda & = 1 \\
\beta &\in \{0.01, 0.05, \underline{0.25}, 0.5\} \\
\text{epochs} & =10
\end{aligned}$ \\
\hline
Average &
$\begin{aligned}
\varepsilon_{\text{win}} &\in [\underline{-2},-1] \\
\varepsilon_{\text{loose}} &\in [-40,\underline{-20}] \\
\text{epochs} = 10
\end{aligned}$ \\
\hline
Pointwise &
$\begin{aligned}
\varepsilon_{\text{win}} &\in [\underline{-2},-1]\\
\varepsilon_{\text{loose}} &\in [\underline{-20},-10] \\
\text{epochs} = 10
\end{aligned}$ \\
\hline
Penalty &
$\begin{aligned}\lambda_0 &\in \{0.1,\underline{1},10,100\}\\
\text{epochs} &= 10\end{aligned}$ \\
\bottomrule
\end{tabular}
\caption{Per-method hyperparameter search spaces for the function calling task. Results correspond to the best performing parameters with respect to validation accuracy.}\label{tab:when2method_hyperparams}
\end{table}


\subsection{Re-ranking with Length Optimization}

Reranking in query–passage tasks orders a candidate set so the passage most relevant to the query appears first, e.g.,

\begin{center}
\small
\setlength{\parskip}{4pt}
\noindent\textbf{Query:} \texttt{what is sec?}

\noindent
\textbf{Positive passage:}
The Securities and Exchange Commission (SEC) is a U.S. government agency that oversees securities transactions, activities of financial professionals an\ldots

\noindent\textbf{Negative passages (9):}
\begin{quote}
\begin{enumerate}\itemsep2pt
    \item Student-athletes to represent SEC at NCAA Convention\ldots
    \item The SEC has four major divisions\ldots Division of Corporation Finance\ldots
    \item Securities and Exchange Commission (SEC)\ldots five commissioners\ldots
    \item Historical Examples\ldots assertion made in sec.\ldots
    \item \ldots filing the papers for BH Global Aviation with the sec\ldots
    \item The sec disclosure filed last month\ldots \$100 million under management\ldots
    \item Trig Cheat Sheet\ldots \textit{sec} $= \frac{\text{hypotenuse}}{\text{adjacent}}$\ldots
    \item SEC Network's Jimmy Dykes\ldots
    \item The FEC has held\ldots organized under sec.\ldots
\end{enumerate}
\end{quote}
\end{center}

Each instance contains one correct passage (the positive) and \(n\) related but incorrect passages (negatives). When the positive is placed first, shorter negatives are preferable to longer ones: in RAG, this reduces token usage. Because correctness is non-negotiable, we enforce the score for the positive passage to be larger than any of the negatives plus a tolerance:

\[
\ell(\pi_\theta;q,p,n_i)\;:=\; \pi_\theta(p \mid q) - \pi_\theta(n_i \mid q) \;\geq\;  \epsilon,
\]

resulting in  the single positive first being ranked first. Once that is guaranteed, we minimize token length as a secondary objective. Because ranking metrics are non-differentiable with respect to the model scores, we optimize a LambdaLoss-style surrogate. The goal of this surrogate is to encourage shorter passages to be ranked higher, while giving larger weight to swaps that would have a larger effect near the top of the ranking. For a candidate set \((d_1,\ldots,d_N)=(p,n_1,\ldots,n_n)\), let
\[
    s_i = \pi_\theta(d_i \mid q)
\]
be the model score assigned to candidate \(d_i\), and let \(L_i\) denote its token length. We define a length-based gain
\[
    g_i = \big(\max_k L_k - L_i\big)_+,
\]
so that shorter passages receive larger gain. Given the current model scores, let \(r_i\) be the rank of candidate \(d_i\), with rank \(1\) corresponding to the highest score. We use the truncated DCG discount
\[
    D_K(r) = \frac{\mathbf{1}\{r \le K\}}{\log_2(r+1)}.
\]

For each pair of candidates \((i,j)\), we define the target preference
\[
    y_{ij} = \operatorname{sign}(g_i-g_j),
\]
and ignore pairs with \(y_{ij}=0\). The pair is weighted by the normalized change in DCG@\({K}\) that would result from swapping the two candidates at their current ranks:
\[
    w_{ij}
    =
    \frac{|\Delta \mathrm{DCG@}K(i,j)|}{\mathrm{IDCG@}K},
\]
where
\[
\begin{aligned}
    \Delta \mathrm{DCG@}K(i,j)
    &=
    \big[g_i D_K(r_j) + g_j D_K(r_i)\big]
    -
    \big[g_i D_K(r_i) + g_j D_K(r_j)\big],
\end{aligned}
\]
and \(\mathrm{IDCG@}K\) is the ideal DCG@\({K}\) obtained by sorting candidates in decreasing order of length-based gain.

The resulting LambdaLoss objective is the normalized weighted pairwise logistic loss
\[
    \Lambda_{\mathrm{loss}}(\{s_i\}_{i=1}^N)
    =
    \frac{
    \sum_{i<j:\, y_{ij}\neq 0}
    w_{ij}
    \log\!\left(1+\exp\!\left[-y_{ij}(s_i-s_j)\right]\right)
    }{
    \sum_{i<j:\, y_{ij}\neq 0} w_{ij}
    }.
\]
Thus, if \(g_i>g_j\), meaning that candidate \(i\) is shorter than candidate \(j\), the loss encourages \(s_i>s_j\). The weights concentrate the learning signal on pairwise swaps that would most affect the top-\(K\) length-aware ranking. We set $K=3$ for all experiments as shown in Table~\ref{tab:reranker-shared-hparams}.

\paragraph{Baselines.}
We compare against three relevance-only baselines. These baselines are trained only to maximize query--passage relevance, but they do not include an explicit length objective. We train a monoBERT reranker following the pointwise reranking stage of the multi-stage document ranking pipeline of~\citet{nogueira2019multi} on our dataset. For each query--candidate pair \((q,d)\), the model is trained with the standard binary cross-entropy loss
\[
    \mathcal{L}_{\mathrm{monoBERT}}(\theta)
    =
    - \sum_{(q,d,y)}
    \Big[
        y \log \sigma(s_\theta(q,d))
        +
        (1-y)\log\big(1-\sigma(s_\theta(q,d))\big)
    \Big],
\]
where \(s_\theta(q,d)\in\mathbb{R}\) is the cross-encoder score and \(\sigma\) is the sigmoid function. At inference time, each candidate passage in the set \(\{p,n_1,\ldots,n_m\}\) is scored independently against the query, and the passages are sorted in decreasing order of \(s_\theta(q,d)\). This baseline measures the performance of a standard reranker trained on our 55k-example subset, but optimized only for relevance.

Moreover, we evaluate on the off-the-shelf rerankers \textit{BGE-large}~\cite{bge_embedding} and \textit{ms-marco-MiniLM}~\cite{hf_cross_encoder_msmarco_minilm_l12_v2}. This models were trained on the MS MARCO Passage Ranking task, i.e., the same large-scale passage-ranking benchmark from which our 55k-example subset is drawn.

\paragraph{Dataset.}
We evaluate our framework on a subset of \emph{MS MARCO} v2.1. The corpus contains approximately \(809\text{k}\) instances for training and \(101\text{k}\)  for validation, each represented as \((q,p,\mathbf{n})\) with a query \(q\), a positive passage \(p\), and a variable number of negatives \(\mathbf{n}\). For our experiments, we sub-sample \(55\text{k}\) instances for train and  \(5\text{k}\) for validation. We fix \(n=9\) negatives per query. The subset is selected to \emph{maximize the variance of negative-passage token lengths} within each instance, stressing length-sensitive behavior.

\paragraph{Evaluation.}
We assess performance through rank-dependent metrics that reflect retrieval behavior. 
For relevance, we report \textbf{MRR@10} and \textbf{Hit@3}. Let \(r_i^+\) denote the rank of the positive passage for query \(i\). Then
\begin{equation}
    \mathrm{MRR@10}
    =
    \frac{1}{N}
    \sum_{i=1}^N
    \frac{\mathbf{1}\{r_i^+ \leq 10\}}{r_i^+},
    \qquad
    \mathrm{Hit@3}
    =
    \frac{1}{N}
    \sum_{i=1}^N
    \mathbf{1}\{r_i^+ \leq 3\}.
\end{equation}

MRR@10 rewards systems that rank the positive passage very early, while Hit@3 measures the fraction of queries for which the positive appears in the top three.

For efficiency, we report a discounted length rank score metric \textbf{LenRank@10} and \textbf{AvgLength@3}. 
Let \(\sigma_i(r)\) be the index of the negative passage ranked in position \(r\) among the negatives for query \(i\), and let \(L_{i,j}\) denote the token length of passage \(j\). We compute
\begin{equation}\label{apx:len-rank}
    \mathrm{LenRank@10}
    =
    \frac{1}{N}
    \sum_{i=1}^N
    \sum_{r=1}^{10}
    \frac{L_{i,\sigma_i(r)}}{r},
\end{equation}

where lower values indicate that shorter negative passages are ranked higher. Finally,
\begin{equation}
    \mathrm{AvgLength@3}
    =
    \frac{1}{3N}
    \sum_{i=1}^N
    \sum_{r=1}^{3}
    L_{i,\tau_i(r)},
\end{equation}

where \(\tau_i(r)\) is the passage ranked in position \(r\) among all candidates for query \(i\). Lower AvgLength@3 indicates lower token usage in the passages that would be passed to the downstream RAG system. Together, MRR@10 and Hit@3 measure relevance quality, while the length-weighted score and AvgLength@3 measure the efficiency of the retrieved context. All metrics are reported on the validation set.

\begin{table}[h!]
\caption{Shared hyperparameters for the re-ranking experiments. All methods use the same ModernBERT-base cross-encoder backbone, candidate-set size, input truncation length, optimization budget, and batching configuration.}
\label{tab:reranker-shared-hparams}
\centering
\begin{tabular}{ll}
\toprule
\textbf{Category} & \textbf{Parameters} \\
\midrule
Task / Model & $\text{base model} = \text{ModernBERT-base}$ \\
             & $\text{max\_length} = 512$ \\
             & $\text{num\_negatives} = 9$ \\
             & $\text{top\_k} = 3$ \\
             & $k_\lambda = 3$ \\
\hline
Training     & $\text{num\_train\_epochs} = 10$ \\
             & $\text{learning\_rate} = 2 \times 10^{-5}$ \\
             & $\text{lr\_scheduler\_type} = \text{linear}$ \\
             & $\text{warmup\_ratio} = 0.1$ \\
             & $\text{weight\_decay} = 0$ \\
             & $\text{gradient\_accumulation\_steps} = 8$ \\
             & $\text{train batch size} = 8$ \\
             & $\text{precision} = \text{bf16}$ \\
             & $\text{dropout} = 0.05$ \\
\bottomrule
\end{tabular}
\end{table}

\paragraph{Hyperparameters.} All re-ranking experiments use the same dataset and training configuration. In particular, we use ModernBERT-base as the shared cross-encoder backbone, truncate query--passage inputs to 512 tokens, and construct each instance with one positive passage and nine negatives. Table~\ref{tab:reranker-shared-hparams} shows the complete list of hyperparameters used.


\subsection{Instruction Following with Safety Refusal}

Instruction-following chat assistants risk complying with malicious or harmful user requests. 
%
Most mitigation strategies typically fall into two families: (i) training-time alignment that teaches the model to refuse unsafe instructions( e.g.~\cite{bai2022training, dai2024safe,huang2024one, wachi2024stepwise, kim2025safedpo}), and (ii) inference-time filtering that blocks unsafe inputs before they reach the LM~\cite{sharma2025constitutional-guardrail, zeng2025shieldgemma-guardrail, sreedhar2025safety-guardrail}. We adopt the former and cast safety as a constraint with helpfulness as the optimization objective.

Moreover, an unintended consequence of safety tuning is that models can start refusing safe queries, phenomenon known as exaggerated safety or over-refusal~\cite{xstest, overefusal}. Therefore, we also add a constraint to avoid refusals on safe prompts. 

Let $\pi_{\theta_{\mathrm{ref}}}$ be a fixed reference LM (e.g., the pretrained model). We denote the distribution of benign instructions paired with helpful responses $D_{
H}$, and $D_{U}$ the distribution of unsafe prompts. We tackle the pointwise constrained learning problem:
\begin{align}
\min_{\theta \in \Theta}. \quad & \EE_{(\bx_h,\by_h)\sim\mathcal{D}_{\text{H}}}
\Big[ D_{KL}\left(\pi_\theta(\by_h \mid \bx_h)\|\pi_{\text{ref}}(\by_h \mid \bx_h)\right) \Big] \label{eq:csft_obj}\\
\text{s.to} \quad & \pi_\theta(\by_r\mid \bx_u) \geq \epsilon_{\text{U}} \quad
x_u \sim \mathcal{D}_{\text{U}}-\text{a.e.}\label{eq:csft_constr1}\\
& \pi_\theta(\by_r\mid \bx_h) \leq \epsilon_\text{H} \quad
x_h \sim \mathcal{D}^X_\text{H}-\text{a.e.}\label{eq:csft_constr2}
\end{align}
The objective in~\eqref{eq:csft_obj} is standard KL divergence on benign data. The constraints bound the probability of refusal responses:  for unsafe content we require this probability to be high~\eqref{eq:csft_constr1}, and for every safe prompt we require the probability to be low~\eqref{eq:csft_constr2}. For numerical stability reasons we impose this constraints not directly on probabilities but on length normalized log likelihood, i.e., $ \frac{1}{\left|y_r\right|} \log \pi_\theta\left({y_r} \mid x\right)$, and set the constraint level to $\epsilon^\prime = \frac{1}{\left|y_r\right|} \log \epsilon$, where epsilon denotes the probability tolerances in~\eqref{eq:csft_constr1}-\eqref{eq:csft_constr2}.

\paragraph{Dataset.}
We use the Alpaca-long1k~\cite{alpaca-long1k} dataset for instruction following. After exploratory experiments that revealed lack of generalization for constraint distributions for the PKU-SafeRLHF-30k dataset~\cite{dai2024safe} (see section~\ref{app:results-fail}), we constructed the safety tuning dataset from the larger Beavertails~\cite{ji2024beavertails} dataset, which provides 330k human-annotated prompt-response pairs across 14 safety categories. To ensure high label quality and mitigate annotation noise, we retained only those samples evaluated by more than two crowd-workers where all annotators were in unanimous agreement. From this subset, we sampled 4,200 instances from each of the three most frequent categories:
\begin{itemize}
    \item \textbf{Violence, Aiding and Abetting, Incitement:} Content involving physical harm, threats, or violent behavior. This includes aiding and abetting (supporting or encouraging illegal activities) and incitement (provoking or stirring up harmful or violent actions).
    \item \textbf{Financial Crime, Property Crime, Theft:} Illegal activities leading to financial loss, such as fraud, embezzlement, money laundering, and bribery. This also encompasses property crimes and the non-consensual taking of property.
    \item \textbf{Privacy Violation:} Actions that invade individual privacy, including unauthorized access or disclosure of personal data, stalking, doxing, or any intrusion into personal life without consent.
\end{itemize}
We hold-out 600 samples, 200 from each category, and 4k samples from the Alpaca dataset~\cite{alpaca} for validation.

\paragraph{Evaluation} We evaluate instruction following using Alpaca-eval~\cite{alpaca_eval} with default hyperparameters. We compare the generations from the finetuned and base model using GPT-5.1 as a judge. We sample the model at the end of training on the held-out safety prompts setting maximum new tokens to $512$ and using the transformers library's\cite{transformers-library} default parameters for sampling (\texttt{temperature=1.0, top\_p=1.0, top\_k=50, repetition\_penalty=1.0, length\_penalty=1.0}). We label a response as a refusal if the generated output begins with a predefined set of refusal-related phrases  ``i’m sorry'', ``i cannot'', ``i do not support'', ``i am programmed to be neutral''. 

\textbf{Hyperparameters} We first finetune the Llama-7b base model on Alpaca longest 1k samples~\cite{alpaca-long1k} minimizing cross entropy via standard empirical-risk minimisation with LoRA adapters ($r=64$, $\alpha=128$, dropout$=0.05$) for 5 epochs using AdamW~\cite{adamw} at learning rate 5e-4 with a linear schedule, 10\% warm-up steps, 0.05 weight decay, and an effective batch size of 32. The resulting finetuned Llama model is used as a base model for all subsequent experiments.

\begin{table}[h!]
\caption{Shared training and model parameters across all runs in the safety-tuning task.}\label{tab:safe-shared_hyperparams}
\centering
\begin{tabular}{l|l}
\toprule
\textbf{Category} & \textbf{Parameters} \\
\midrule
Training &
$\begin{aligned}
\text{epochs} &= 10 \\
\text{learning rate} &= 5 \times 10^{-4} \\
\text{batch size} &= 32 \\
\text{lr scheduler} &= \text{cosine} \\
\text{warmup ratio} &= 0.1 \\
\text{weight decay} &= 0.01 \\
\text{optimizer} &= \text{AdamW (fused)} \\
\text{max grad norm} &= 1.0 \\
\text{precision} &= \text{bf16} \\
\text{max length} &= 2048
\end{aligned}$ \\
\hline
LoRA &
$\begin{aligned}
r &= 64 \\
\text{lora\_alpha} &= 128 \\
\text{lora\_dropout} &= 0.05 \\
\text{target modules} &= \{q,k,v,o\}\text{\_proj}
\end{aligned}$ \\
\bottomrule
\end{tabular}
\end{table}

\begin{table}[h!]
\caption{Per method hyperparameters for the safety tuning task}\label{tab:safe-method_hyperparams}
\centering
\begin{tabular}{l|l}
\toprule
\textbf{Method} & \textbf{Hyperparameters} \\
\midrule
pen &
$\begin{aligned}
\lambda_0 &= 5 
\end{aligned}$ \\
\hline
Safe-DPO &
$\begin{aligned}
\Delta &\in\{2, 5, 10\}
\end{aligned}$ \\
\hline
avg &
$\begin{aligned}
\alpha &= 1000 \\
\varepsilon_\text{U} &\in\{ 0.08, 0.16, -0.24, -0.28, -0.32 \}  \\
\varepsilon_\text{H} &= -2 \\
\eta_{\text{dual}} &= 1 \\
\end{aligned}$ \\
\hline
point &
$\begin{aligned}
\alpha &= 1000 \\
\varepsilon_\text{U} &\in\{ 0.08, 0.16, -0.24, -0.32 \}  \\
\varepsilon_\text{H} &= -2 \\
\eta_{\text{dual}} &= 1 \\
\end{aligned}$ \\
\hline
relax &
$\begin{aligned}
\alpha &= 1000 \\
\varepsilon_\text{U} &=0.08  \\
\varepsilon_\text{H} &= -2 \\
\eta_{\text{dual}} &= 1 \\
\beta &\in \{ 5, 10, 20\}
\end{aligned}$\\
\bottomrule
\end{tabular}
\end{table}
\section{Experimental Results}\label{app:results}

\subsection{Function calling preferences}\label{app:results-func}
Table~\ref{tab:when2call_results} presents the evaluation results on the When2Call benchmark across all  models and methods. Firstly, all alignment techniques substantially improve upon the pre-trained models, which struggle with asking follow-up questions and recognizing when they are unable to answer, often resulting in severe hallucination rates (e.g., up to 57.8\% for xLAM-2). While standard preference alignment techniques (such as DPO and SimPO) drastically reduce these hallucinations and improve overall accuracy, they still exhibit uneven performance across the categories of the correct labels in terms of F1 scores. Explicitly, both models show lower performance on ``Follow-up question''.
Our approach effectively reduce this disparities. For example,  pointwise constraints achieves an 8\% improvement Follow-up F1 score with respect to SimPO for the xLAM model. Moreover, the relaxed variant (\textbf{relax}) successfully matches or slightly improves upon pointwise constraints across all categories. This is more prominent for the xLAM-2-1b-fc-r model, where the relaxed method not only improves in difficult categories, but also achieves the highest overall scores (Macro F1 and a Length Normalized Accuracy), while maintaining a perfect (0.0\%) hallucination rate.

\begin{sidewaystable}[]
\caption{Function calling (When2Call) results across models and methods. Values are shown as percentages. Baseline corresponds to the corresponding pre-trained model. We include results for Mistral-Nemo-Minitron sourced from~\cite{when2call} as a reference. The best and second-best scores per model are \textbf{bolded} and \underline{underlined}, respectively.}
\label{tab:when2call_results}
\centering
\small
\begin{tabular}{llcccccc}
\toprule
\multirow{2}{*}{Model} & \multirow{2}{*}{Method} & \multicolumn{4}{c}{F1 ($\uparrow$)} & \multirow{2}{*}{Length Norm Acc ($\uparrow$)} & \multirow{2}{*}{Hallucination ($\downarrow$)} \\
\cmidrule(lr){3-6}
 & & Tool call ($\uparrow$) & Follow-up question ($\uparrow$) & Unable to answer ($\uparrow$) & Macro average ($\uparrow$) & & \\
\midrule
\multirow{2}{*}{Mistral-NeMo-Minitron} & Base & 65.8 & 33.3 & 28.3 & 42.5 & 49.1 & 19.0 \\
 & RPO & 74.8 & 63.6 & 71.1 & 69.8 & 70.0 & 1.2 \\
\midrule
\multirow{9}{*}{Llama-3.2-1B-Instruct} & Base & 57.3 & 8.3 & 8.5 & 24.7 & 47.6 & 52.3 \\
 & DPO & 75.4\textcolor{gray}{$\pm$\,0.2} & 72.4\textcolor{gray}{$\pm$\,0.7} & 72.5\textcolor{gray}{$\pm$\,1.6} & 73.4\textcolor{gray}{$\pm$\,0.8} & 78.7\textcolor{gray}{$\pm$\,0.4} & 5.3\textcolor{gray}{$\pm$\,1.6} \\
 & SimPO & 78.0\textcolor{gray}{$\pm$\,0.1} & 72.3\textcolor{gray}{$\pm$\,0.8} & 76.7\textcolor{gray}{$\pm$\,1.6} & 75.7\textcolor{gray}{$\pm$\,0.6} & 79.1\textcolor{gray}{$\pm$\,0.4} & 3.4\textcolor{gray}{$\pm$\,0.2} \\
 & CPO & \textbf{80.8}\textcolor{gray}{$\pm$\,1.6} & 72.9\textcolor{gray}{$\pm$\,1.5} & \underline{86.7}\textcolor{gray}{$\pm$\,0.9} & \textbf{80.1}\textcolor{gray}{$\pm$\,1.3} & 76.8\textcolor{gray}{$\pm$\,4.9} & \underline{0.4}\textcolor{gray}{$\pm$\,0.0} \\
 & Cal-DPO & 78.8\textcolor{gray}{$\pm$\,1.4} & 72.8\textcolor{gray}{$\pm$\,2.3} & 81.9\textcolor{gray}{$\pm$\,0.6} & 77.8\textcolor{gray}{$\pm$\,1.3} & 78.1\textcolor{gray}{$\pm$\,1.5} & 1.0\textcolor{gray}{$\pm$\,0.2} \\
 & pen & 74.3\textcolor{gray}{$\pm$\,0.6} & 65.1\textcolor{gray}{$\pm$\,1.3} & 71.9\textcolor{gray}{$\pm$\,0.0} & 70.4\textcolor{gray}{$\pm$\,0.7} & 74.4\textcolor{gray}{$\pm$\,0.8} & 1.7\textcolor{gray}{$\pm$\,0.6} \\
 & avg & 77.1\textcolor{gray}{$\pm$\,2.6} & 60.2\textcolor{gray}{$\pm$\,12.1} & \textbf{87.1}\textcolor{gray}{$\pm$\,1.0} & 74.8\textcolor{gray}{$\pm$\,4.7} & 71.7\textcolor{gray}{$\pm$\,5.4} & 0.5\textcolor{gray}{$\pm$\,0.2} \\
 & point & 79.3\textcolor{gray}{$\pm$\,0.4} & \textbf{77.2}\textcolor{gray}{$\pm$\,0.3} & 79.1\textcolor{gray}{$\pm$\,0.4} & 78.5\textcolor{gray}{$\pm$\,0.1} & \underline{79.5}\textcolor{gray}{$\pm$\,0.4} & \textbf{0.2}\textcolor{gray}{$\pm$\,0.2} \\
 & relax & \underline{79.6}\textcolor{gray}{$\pm$\,0.2} & \underline{76.9}\textcolor{gray}{$\pm$\,0.7} & 80.0\textcolor{gray}{$\pm$\,0.4} & \underline{78.8}\textcolor{gray}{$\pm$\,0.4} & \textbf{80.0}\textcolor{gray}{$\pm$\,0.3} & 0.6\textcolor{gray}{$\pm$\,0.2} \\
\midrule
\multirow{9}{*}{xLAM-2-1b-fc-r} & Base & 57.2 & 11.2 & 12.1 & 26.8 & 42.9 & 57.8 \\
 & DPO & 78.1\textcolor{gray}{$\pm$\,1.1} & 75.3\textcolor{gray}{$\pm$\,1.2} & 82.4\textcolor{gray}{$\pm$\,0.9} & 78.6\textcolor{gray}{$\pm$\,1.1} & 80.5\textcolor{gray}{$\pm$\,0.5} & 1.8\textcolor{gray}{$\pm$\,0.5} \\
 & SimPO & 74.0\textcolor{gray}{$\pm$\,2.0} & 72.1\textcolor{gray}{$\pm$\,0.3} & 81.4\textcolor{gray}{$\pm$\,0.3} & 75.9\textcolor{gray}{$\pm$\,0.8} & 78.2\textcolor{gray}{$\pm$\,1.6} & \underline{0.3}\textcolor{gray}{$\pm$\,0.2} \\
 & CPO & 82.0\textcolor{gray}{$\pm$\,5.6} & 70.5\textcolor{gray}{$\pm$\,9.4} & 83.4\textcolor{gray}{$\pm$\,5.3} & 78.7\textcolor{gray}{$\pm$\,6.8} & 80.3\textcolor{gray}{$\pm$\,4.6} & 4.5\textcolor{gray}{$\pm$\,6.4} \\
 & Cal-DPO & 80.8\textcolor{gray}{$\pm$\,0.5} & 76.6\textcolor{gray}{$\pm$\,0.2} & 84.9\textcolor{gray}{$\pm$\,0.3} & 80.8\textcolor{gray}{$\pm$\,0.1} & 81.2\textcolor{gray}{$\pm$\,0.3} & \textbf{0.0}\textcolor{gray}{$\pm$\,0.0} \\
 & pen & 79.5\textcolor{gray}{$\pm$\,0.7} & 71.6\textcolor{gray}{$\pm$\,0.4} & 81.5\textcolor{gray}{$\pm$\,0.7} & 77.6\textcolor{gray}{$\pm$\,0.5} & 78.9\textcolor{gray}{$\pm$\,0.7} & 0.9\textcolor{gray}{$\pm$\,0.5} \\
 & avg & \textbf{84.1}\textcolor{gray}{$\pm$\,0.9} & 75.8\textcolor{gray}{$\pm$\,0.5} & \underline{86.2}\textcolor{gray}{$\pm$\,0.6} & 82.0\textcolor{gray}{$\pm$\,0.5} & 82.1\textcolor{gray}{$\pm$\,0.6} & \textbf{0.0}\textcolor{gray}{$\pm$\,0.0} \\
 & point & 83.6\textcolor{gray}{$\pm$\,0.3} & \underline{80.9}\textcolor{gray}{$\pm$\,0.4} & 86.1\textcolor{gray}{$\pm$\,0.4} & \underline{83.5}\textcolor{gray}{$\pm$\,0.2} & \underline{83.7}\textcolor{gray}{$\pm$\,0.2} & \textbf{0.0}\textcolor{gray}{$\pm$\,0.0} \\
 & relax & \underline{84.0}\textcolor{gray}{$\pm$\,0.3} & \textbf{81.1}\textcolor{gray}{$\pm$\,0.4} & \textbf{86.4}\textcolor{gray}{$\pm$\,0.3} & \textbf{83.8}\textcolor{gray}{$\pm$\,0.3} & \textbf{84.1}\textcolor{gray}{$\pm$\,0.1} & \textbf{0.0}\textcolor{gray}{$\pm$\,0.0} \\
\bottomrule
\end{tabular}
\end{sidewaystable}

We also analyze the distributions of length normalized probabilities assigned to the preferred and rejected responses by the Llama3 model in the When2Call test split across all alignment methods. We present the empirical Cumulative Density Function (CDF) for of response probabilities for different subsets of methods to ease readability. Figure~\ref{fig:point-average-pen-cdfs-when2} shows  that the Average and Penalty methods exhibit heavier tails for both the preferred and rejected responses. Figure~\ref{fig:simpo-when2} shows that SimPO decreases the probability of both responses --including the prefered one --. In contrast, the Pointwise and Resilient approaches successfully maintain the probability of the chosen response high. Figure~\ref{fig:cpo-when2} shows that CPO, which mitigates this phenomenon by adding a penalization to increase the probability of the preferred response to the margin loss, yields a heavier tailed distribution akin to average constraints. Similarly, in Figure~\ref{fig:cal-dpo} Cal-DPO also shows heavier tails. Lastly, Figure~\ref{fig:resilient-when2} shows how decreasing the relaxation cost result in lower probabilities for the prefered response and increases the rejected one.

\begin{figure}
    \centering
    \includegraphics[width=0.9\linewidth]{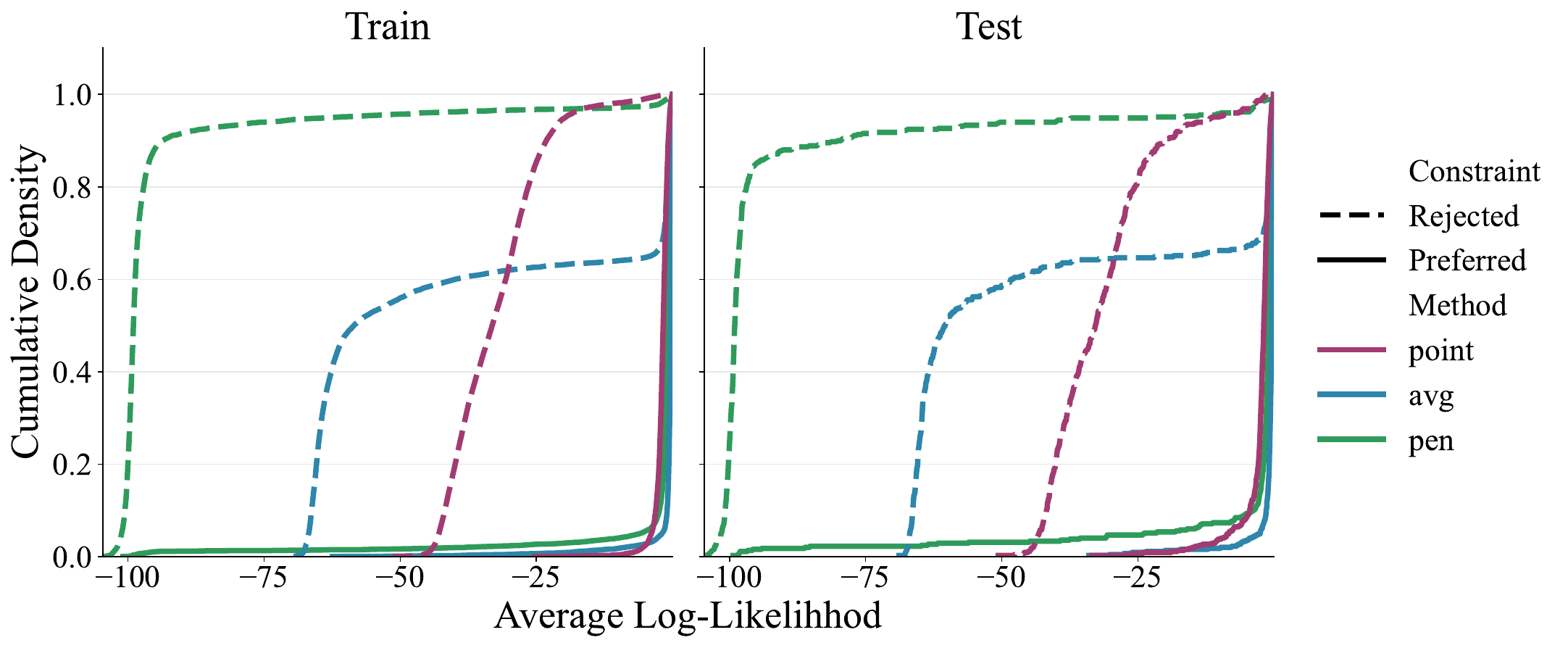}
    \caption{Empirical CDF of the tokenwise average log likelihood, i.e. $\frac{1}{|\bby|}\log(\pi(\bby|\bbx))$, evaluated on Llama-3.2-1B on When2Call. We compare pointwise constraints, average constraints, a fixed penalty with the base model for chosen and rejected responses.\textit{ Pointwise constraints achieve more concentrated probability distributions.}}
    \label{fig:point-average-pen-cdfs-when2}
\end{figure}

\begin{figure}
    \centering
    \includegraphics[width=0.9\linewidth]{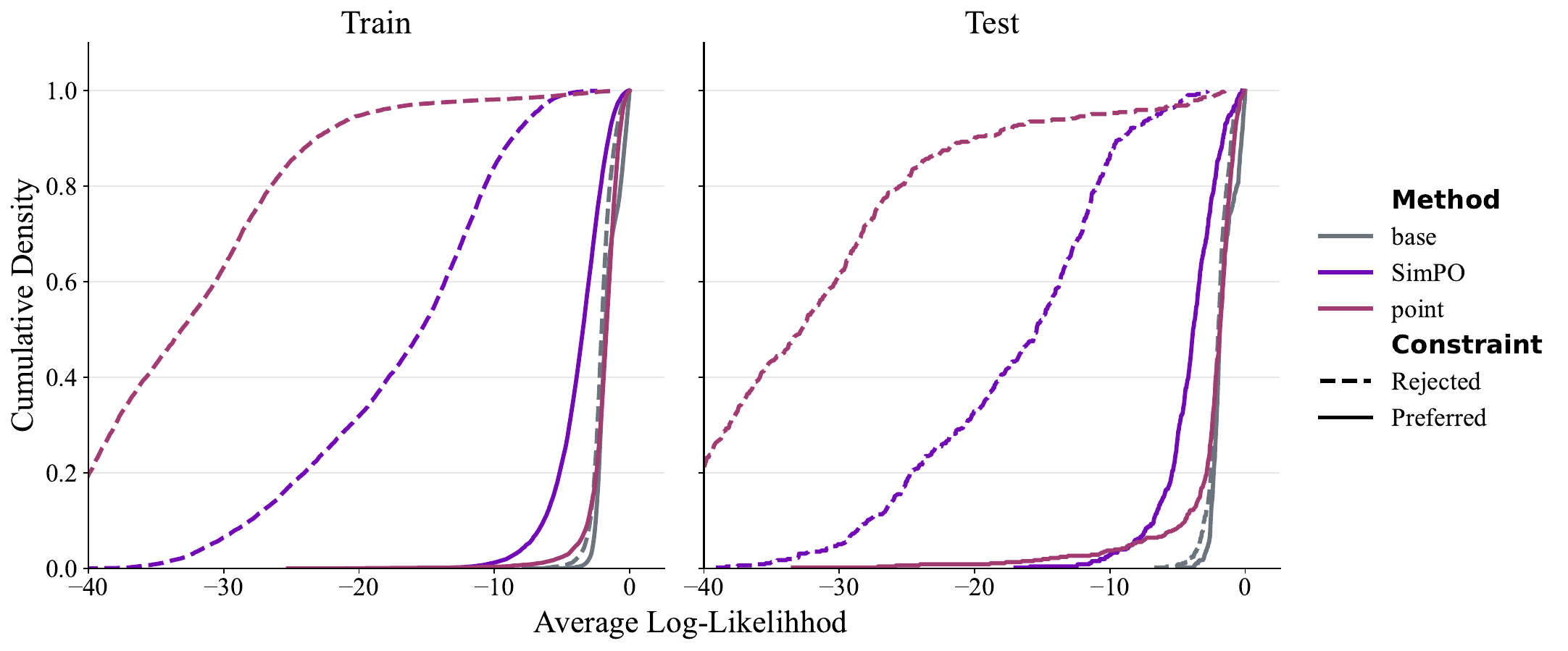}
    \caption{Empirical CDF of the tokenwise average log likelihood, i.e. $\frac{1}{|\bby|}\log(\pi(\bby|\bbx))$, evaluated on Llama-3.2-1B on when2call. We compare pointwise constraints with SimPO for chosen and rejected responses.\textit{ Margin losses reduce the likelihood of prefered responses.}}
    \label{fig:simpo-when2}
\end{figure}

\begin{figure}
    \centering
    \includegraphics[width=0.9\linewidth]{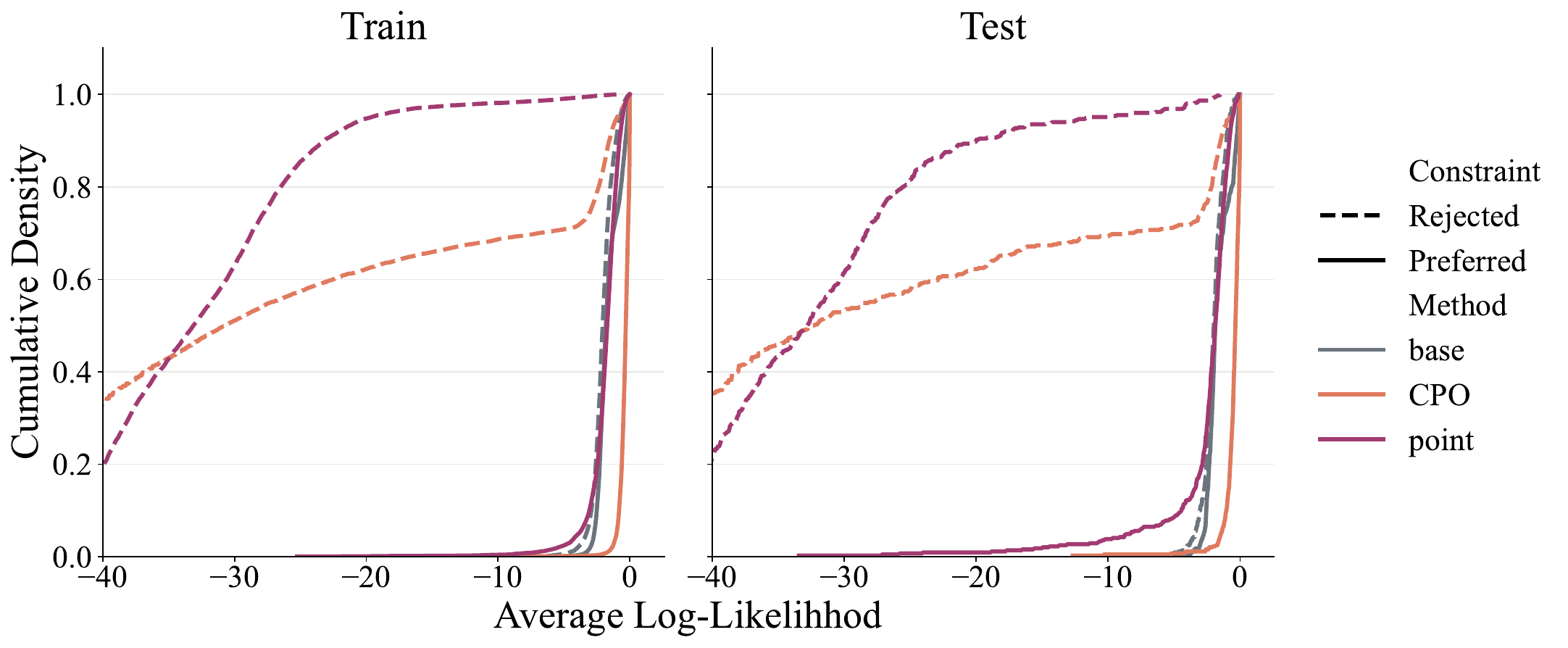}
    \caption{Empirical CDF of the tokenwise average log likelihood, i.e. $\frac{1}{|\bby|}\log(\pi(\bby|\bbx))$, evaluated on Llama-3.2-1B on when2call. We compare pointwise constraints with CPO for chosen and rejected responses.\textit{ Adding average penalties to margin losses (CPO) results in heavier tails.}}
    \label{fig:cpo-when2}
\end{figure}

\begin{figure}
    \centering
    \includegraphics[width=0.9\linewidth]{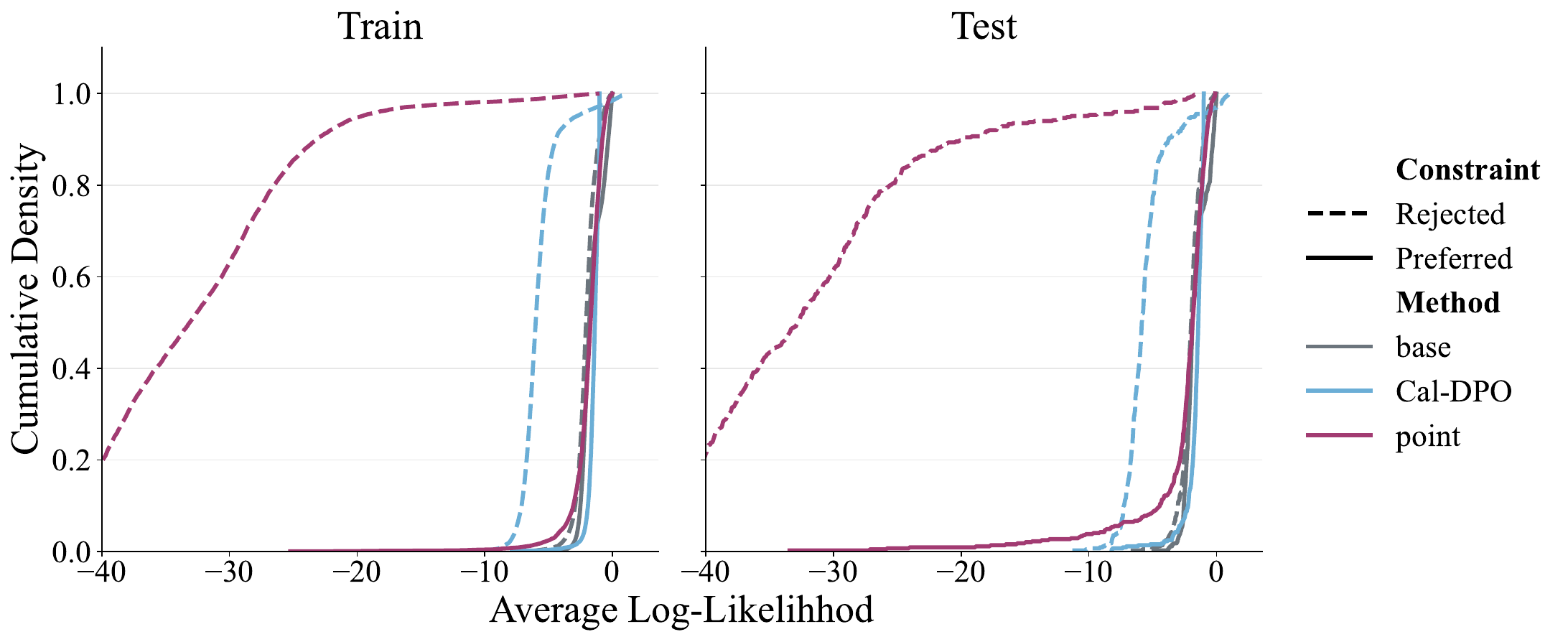}
    \caption{Empirical CDF of the tokenwise average log likelihood, i.e. $\frac{1}{|\bby|}\log(\pi(\bby|\bbx))$, evaluated on Llama-3.2-1B on when2call. We compare pointwise constraints with Cal-DPO for chosen and rejected responses.\textit{ Quadratic penalty methods (Cal-DPO) exhibit heavier tails than our approach.}}
    \label{fig:cal-dpo}
\end{figure}

\begin{figure}
    \centering
    \includegraphics[width=0.9\linewidth]{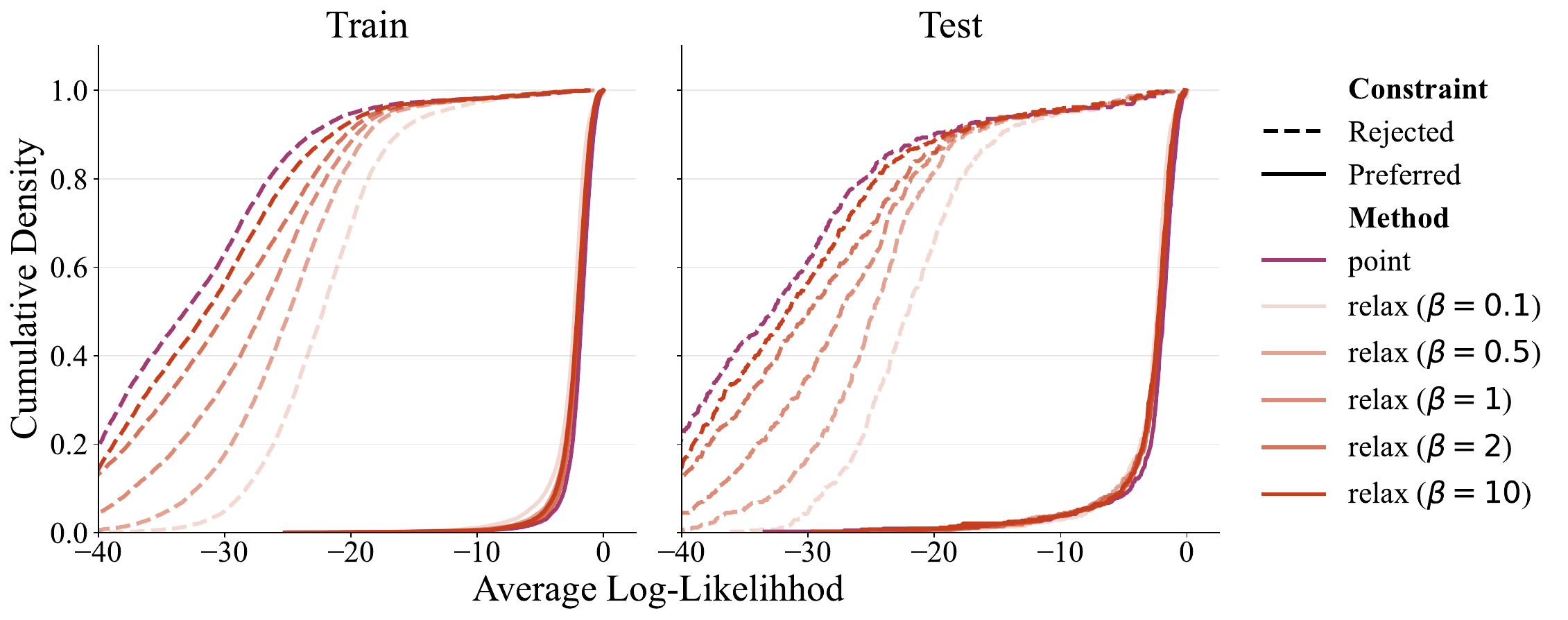}
    \caption{Empirical CDF of the tokenwise average log likelihood, i.e. $\frac{1}{|\bby|}\log(\pi(\bby|\bbx))$, evaluated on Llama-3.2-1B on when2call. We compare pointwise constraints with the relaxation for different costs $\beta$.}
    \label{fig:resilient-when2}
\end{figure}

\subsection{Safety Refusal}\label{app:res-refusal}

We also analyze the distributions of length normalized probabilities assigned to the refusal responses for safe and unsafe prompts by the Llama2-7b model across different alignment methods. 
As shown in figure~\ref{fig:pen-avg-point-cdf-safe}, the fixed penalty (which we set equally for safe and unsafe data) heavily over penalizes refusals for safe data and under-penalizes refusals for unsafe data. Optimizing this scalar weight by imposing average constraints leads to higher probabilities for unsafe prompt refusals, yet with more spread across samples. Lastly, the Safe-DPO baseline exhibits heavy tails in the test set, for both the preferred and rejected responses. Figure~\ref{fig:resilience-cdf-safe} shows that the Resilient approach relaxes the constraints, leading to heavier tailed likelihoods for unsafe refusals, but more concentrated refusal likelihoods for unsafe prompts. 

Table~\ref{tab:app-safety} presents the evaluation metrics for safety and instruction-following capabilities. Firstly, while Safe-DPO can increase the harmlessness of the model, it comes at the cost of a severely degrading the model's instruction-following abilities. Explicitly, for $\Delta=10$ achieves a 78.6\% harmless rate but drastically reduces the AlpacaEval length-controlled win rate to 15.7\%, caused a large KL divergence (2.69) with respect to the pre-trained model. Our approach achieves safety (99.8\% Refusal rate and 96.6\% Harmless) at a considerably cost in terms of KL divergence (0.15) and a 45.0\% win-rate in alpaca-eval with respect to the base model. Lastly decreasing the hyperparameter (e.g., \textbf{relax} $\beta=5$) successfully recovers the win rate to 49.9\%, nearly matching the base model, but it does so at the cost of safety. Although at low safety scores, average constraints also show competitive performance in terms of helpfulness, these models might be unacceptably unsafe to be of any practical value.

\begin{figure}
    \centering
\includegraphics[width=0.9\linewidth]{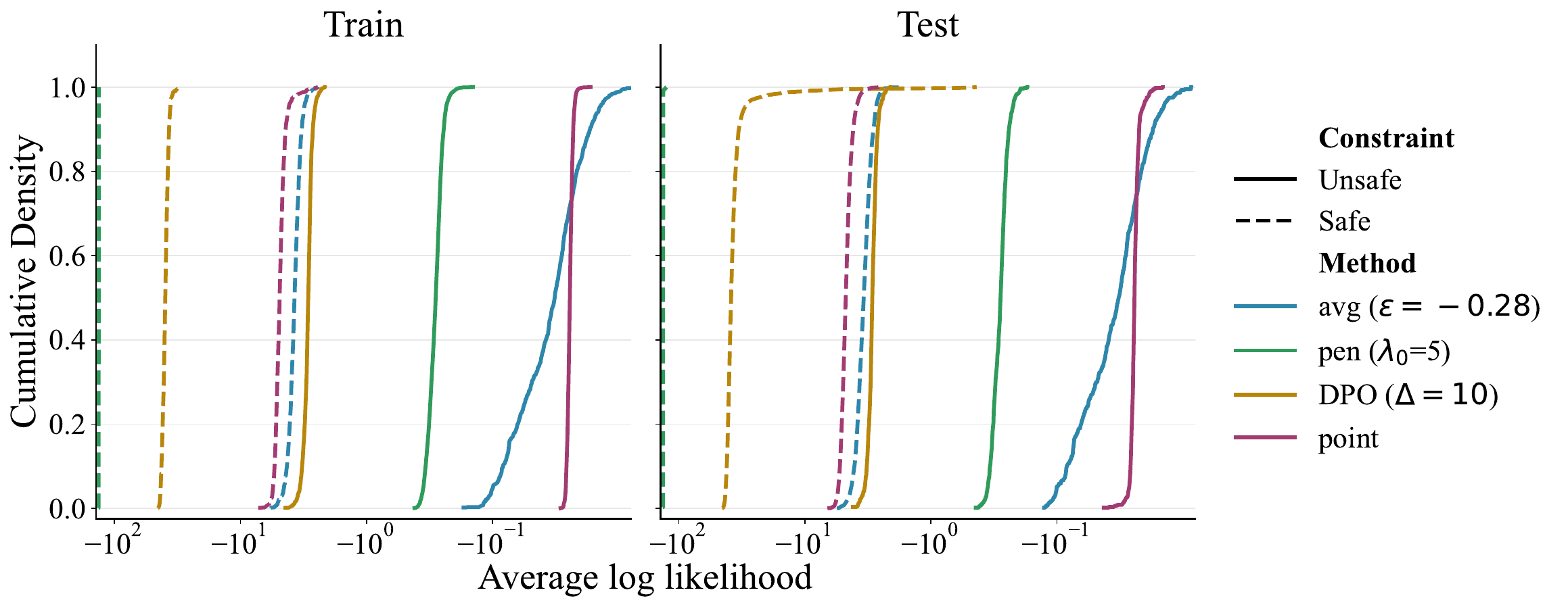}
    \caption{Empirical CDF of the tokenwise average log likelihood, i.e. $\frac{1}{|\bby|}\log(\pi(\bby|\bbx))$, evaluated on Llama-7b on Alpaca and beavertails. We compare pointwise constraints, average constraints a fixed penalty, and Safe-DPO.}
    \label{fig:pen-avg-point-cdf-safe}
\end{figure}

\begin{figure}
    \centering
\includegraphics[width=0.9\linewidth]{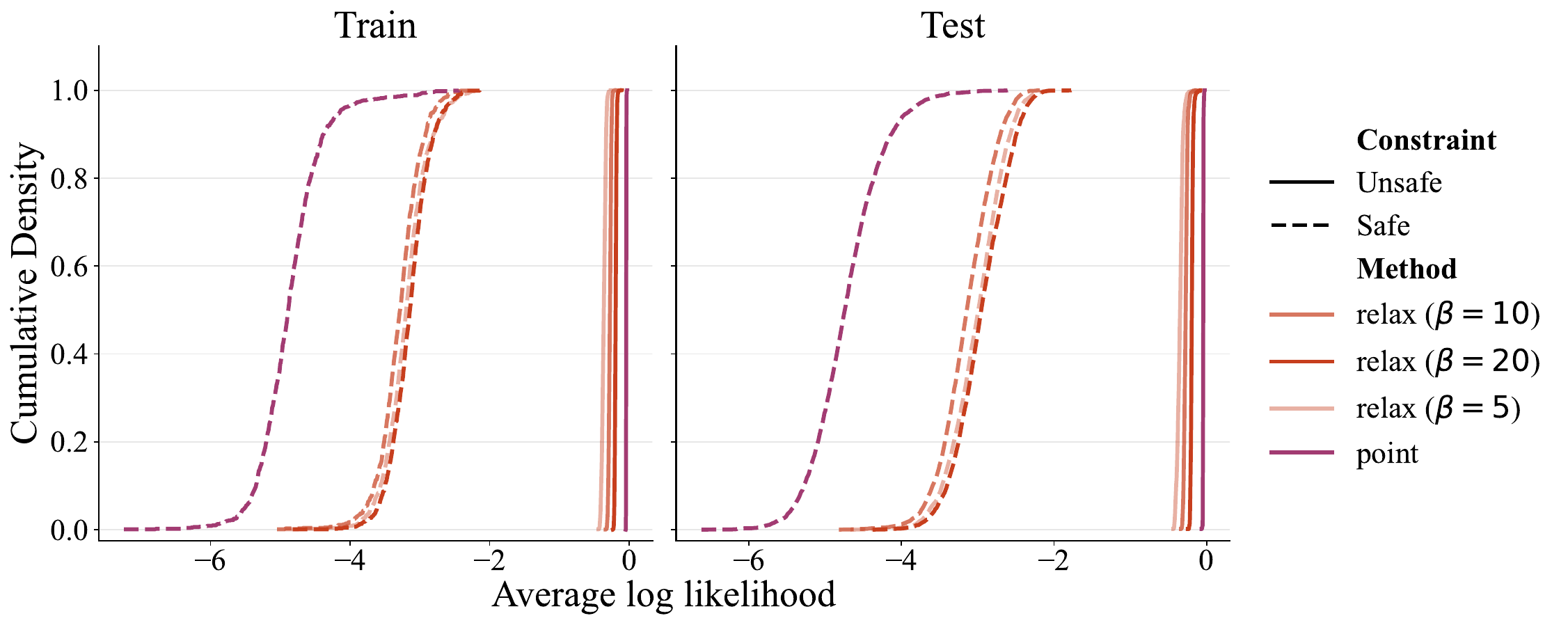}
    \caption{Empirical CDF of the tokenwise average log likelihood, i.e. $\frac{1}{|\bby|}\log(\pi(\bby|\bbx))$, evaluated on Llama2-7b on Alpaca and beavertails. We compare pointwise constraints, and resilience with different values of $\beta$.}
    \label{fig:resilience-cdf-safe}
\end{figure}

\begin{table}[t]
\caption{KL divergence in a held out set of alpaca~\cite{alpaca}, length controlled win rates on AlpacaEval~\cite{alpaca_eval}, refusal rates over harmful prompts from beavertails~\cite{ji2024beavertails} and percentage of responses classified as harmless by \texttt{beaver-7b-unified-cost} from PKU-SafeRLHF~\cite{dai2024safe} evaluation split. The base model is Llama-7B~\cite{touvron2023llama} finetuned on alpaca-1k-longest~\cite{alpaca-long1k}. \emph{Pointwise constraints achieve high refusal rates with minimal degradation in instruction following abilities.}}\label{tab:app-safety}
    \centering
\begin{tabular}{lcccc}
\toprule
 & \multicolumn{1}{c}{Alpaca} & AlpacaEval & BeaverTails & PKU-SafeRLHF \\
 & $D_{\mathrm{KL}}$($\downarrow$) & LC WR($\uparrow$) & Refusal($\uparrow$) & Harmless($\uparrow$) \\
\midrule
Base & 0.00 & 50.0\% & 0.0\% & 29.9\% \\
\midrule
Safe-DPO ($\Delta{=}2$) & 1.35 & 24.9\% & 36.0\% & 40.6\% \\
Safe-DPO ($\Delta{=}5$) & 0.80 & 32.3\% & 48.2\% & 32.6\% \\
Safe-DPO ($\Delta{=}10$) & 2.69 & 15.7\% & 74.7\% & 78.6\% \\
avg ($\epsilon{=}{-}0.32$) & 0.34 & 49.6\% & 39.7\% & 56.8\% \\
avg ($\epsilon{=}{-}0.28$) & 0.41 & 44.2\% & 70.0\% & 88.0\% \\
avg ($\epsilon{=}{-}0.24$) & 0.17 & 41.2\% & 100.0\% & 96.4\% \\
point & 0.15 & 45.0\% & 99.8\% & 96.6\% \\
relax ($\beta{=}20$) & 0.10 & 46.8\% & 73.0\% & 76.6\% \\
relax ($\beta{=}10$) & 0.09 & 48.4\% & 49.3\% & 61.2\% \\
relax ($\beta{=}5$) & 0.09 & 49.9\% & 21.0\% & 45.8\% \\
\bottomrule
\end{tabular}
\end{table}

\subsection{Re-ranking}\label{app:res-reranking}

\begin{sidewaystable}[p]
\caption{
Test-set performance of all methods and hyperparameter configurations for the reranking experiment. 
The column Avg. Top-3 Len. reports the average length of the top-$3$ ranked documents, 
$\mathrm{Hit@3}$ reports whether the relevant document appears among the top-$3$ results, 
$\mathrm{MRR@10}$ reports the mean reciprocal rank within the top-$10$, 
and $\mathrm{LenRank@10}$ reports the length-aware ranking metric used as the training objective. 
The Objective column gives the evaluated objective value, while Slack $\mathrm{CVaR}_{95}$ and Slack Mean summarize, respectively, the tail and average behavior of the constraint slacks on the test set.
}
\label{tab:reranking_all_results}
\centering
\small
\setlength{\tabcolsep}{5pt}
\renewcommand{\arraystretch}{1.15}

\begin{tabular}{clrrrrrrr}
\toprule
\textbf{Method}
& \textbf{Hyperparameter}
& {\textbf{Avg. Top-3 Len. $(\downarrow)$}}
& {\textbf{Hit@3}$(\uparrow)$}
& {\textbf{MRR@10}$(\uparrow)$}
& {\textbf{LenRank@10}$(\downarrow)$}
& {\textbf{Objective}$(\downarrow)$}
& {\textbf{Slack $\mathrm{CVaR}_{95}$}$(\downarrow)$}
& {\textbf{Slack Mean}$(\downarrow)$} \\
\midrule

{\large\textsc{monoBERT}}
& -- & 73.581 & \textbf{0.906} & \underline{0.740} & 202.851 & -- & -- & --\\
\midrule

\multirow{5}{*}{\large\textsc{Avg}}
& $\epsilon = 50$  & 44.842 & 0.616 & 0.505 & 152.539 & 1.697  & 112.286 & -59.065  \\
& $\epsilon = 60$  & 48.036 & 0.686 & 0.543 & 157.043 & 3.588  & 118.515 & -68.477  \\
& $\epsilon = 75$  & 50.240 & 0.731 & 0.591 & 160.206 & 9.055  & 137.722 & -90.809  \\
& $\epsilon = 90$  & 56.463 & 0.805 & 0.638 & 168.411 & 21.963 & 119.366 & -116.963 \\
& $\epsilon = 100$ & 64.203 & 0.812 & 0.637 & 185.447 & 40.305 & 144.588 & -155.235 \\
\midrule

\multirow{5}{*}{\large\textsc{Pen}}
& $\lambda_0 = 0.1$  & \textbf{43.536} & 0.547 & 0.394 & \underline{150.561} & 2.352  & 223.857 & -74.209  \\
& $\lambda_0 = 0.2$  & 48.240 & 0.659 & 0.487 & 157.564 & 7.380  & 235.806 & -91.427  \\
& $\lambda_0 = 0.4$  & 53.133 & 0.744 & 0.571 & 164.371 & 23.262 & 176.521 & -131.843 \\
& $\lambda_0 = 1.0$  & 62.241 & 0.782 & 0.594 & 182.180 & 43.595 & 111.754 & \underline{-162.433} \\
& $\lambda_0 = 10.0$ & 71.531 & 0.813 & 0.619 & 198.403 & 81.484 & 21.426  & \textbf{-185.201} \\
\midrule

\multirow{2}{*}{\large\textsc{Point}}
& $\epsilon = 1$ & 62.078 & 0.883 & 0.727 & 176.850 & 0.544 & \textbf{0.708} & -1.763 \\
& $\epsilon = 3$ & 67.025 & \underline{0.903} & \textbf{0.743} & 186.180 & {1.336} & 1.499 & -3.848 \\
\midrule

\multirow{7}{*}{\large\textsc{Relaxed}}
& $\beta = 0.1$  & \underline{43.719} & 0.59 & 0.458 & \textbf{148.623} & \textbf{0.152} & 2.404 & -1.849 \\
& $\beta = 0.15$ & 47.208 & 0.681 & 0.536 & 152.746 & \underline{0.216} & 2.973 & -2.383 \\
& $\beta = 0.2$  & 48.942 & 0.714 & 0.571 & 155.080 & 0.252 & 2.736 & -2.345 \\
& $\beta = 0.3$  & 51.006 & 0.748 & 0.616 & 158.401 & 0.287 & 3.090 & -2.202 \\
& $\beta = 0.4$  & 53.289 & 0.785 & 0.646 & 161.687 & 0.333 & 2.169 & -2.121 \\
& $\beta = 1.0$  & 56.073 & 0.830 & 0.678 & 166.701 & 0.402 & \underline{1.022} & -1.789 \\
& $\beta = 10.0$ & 58.828 & 0.846 & 0.706 & 171.021 & 0.454 & 1.837 & -2.085 \\

\bottomrule
\end{tabular}
\end{sidewaystable}

Table~\ref{tab:reranking_all_results} reports the complete test-set results for the metrics detailed above. For each method and hyperparameter configuration, we evaluate both relevance performance, through $\mathrm{Acc@3}$ and $\mathrm{MRR@10}$, and length-aware behavior, through the average length of the top-$3$ retrieved documents and $\mathrm{LenRank@10}$. We also include the corresponding objective value and slack statistics, which summarize the magnitude and tail behavior of the constraint violations on the test set. Unless otherwise specified, the dual step size is set to $0.1$ for the average, relaxed, and pointwise methods. For the relaxed method, we fix the constraint level to $\epsilon = 3$, while both the relaxed and pointwise methods use augmented Lagrangian parameter $\alpha = 1$. These results provide a detailed view of the trade-offs induced by each training strategy.

\subsection{Failure cases and Negative results}\label{app:results-fail}

Our framework is not oblivious to generalization challenges and data  limitations. We want to highlight two particular failure modes that we have observed in some experimental settings in the spirit of aiding future research.

In some settings we observed that changing the train distribution of losses did not significantly affect the losses for test samples, i.e., lacking generalization. In figure~\ref{fig:fail-safe-no-gen} we show that for a Llama-7b model trained on PKU-SafeRLHF-30k~\cite{dai2024safe} the test empirical cumulative density functions are almost identical whereas differences in train are observed. 

Moreover, in many settings the constraint is a surrogate for some downstream performance metric. This is the case, for example, of log-likelihood ratios between preferred and rejected responses in preference alignment. As a result, sometimes shifts in the distribution of the surrogate do not affect metrics, like accuracy, in the desired way. Figure~\ref{fig:orca-failure} shows the CDF of log likelihood ratios for preference data from a cleaned version of intel orca dpo pairs\footnote{\texttt{\url{https://huggingface.co/datasets/argilla/distilabel-intel-orca-dpo-pairs}}}, for Llama2-7B trained with pointwise constraints and SimPO. Pointwise constraints achieve 5\% more training samples where the likelihood of the prefered response is higher than the rejected response, i.e., classified correctly. On test samples, both methods are identical with respect to the number of correct samples. However, the distributions do differ, and models may exibit different behaviours in other respects e.g. in calibration.

\begin{figure}
    \centering \includegraphics[width=0.9\linewidth]{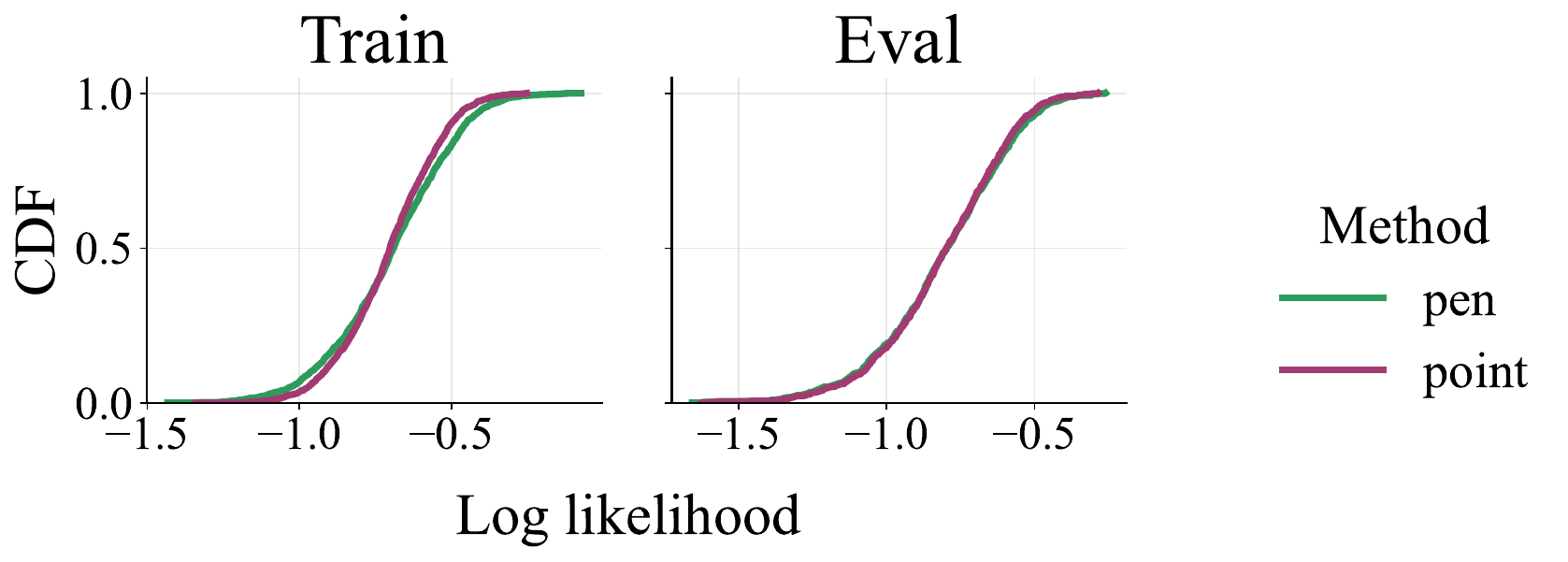}
    \caption{Empirical Cumulative Density Function (CDF) of length normalized likelihood for unsafe responses in the Safe-RLHF~\cite{dai2024safe} dataset for a Llama2-7b in training (left) and test (right) splits. \textit{Although pointwise constraints induce a  distribution shift in training samples it fails to generalize to held-out data.}}
    \label{fig:fail-safe-no-gen}
\end{figure}

\begin{figure}
\centering
\includegraphics[width=0.9\linewidth]{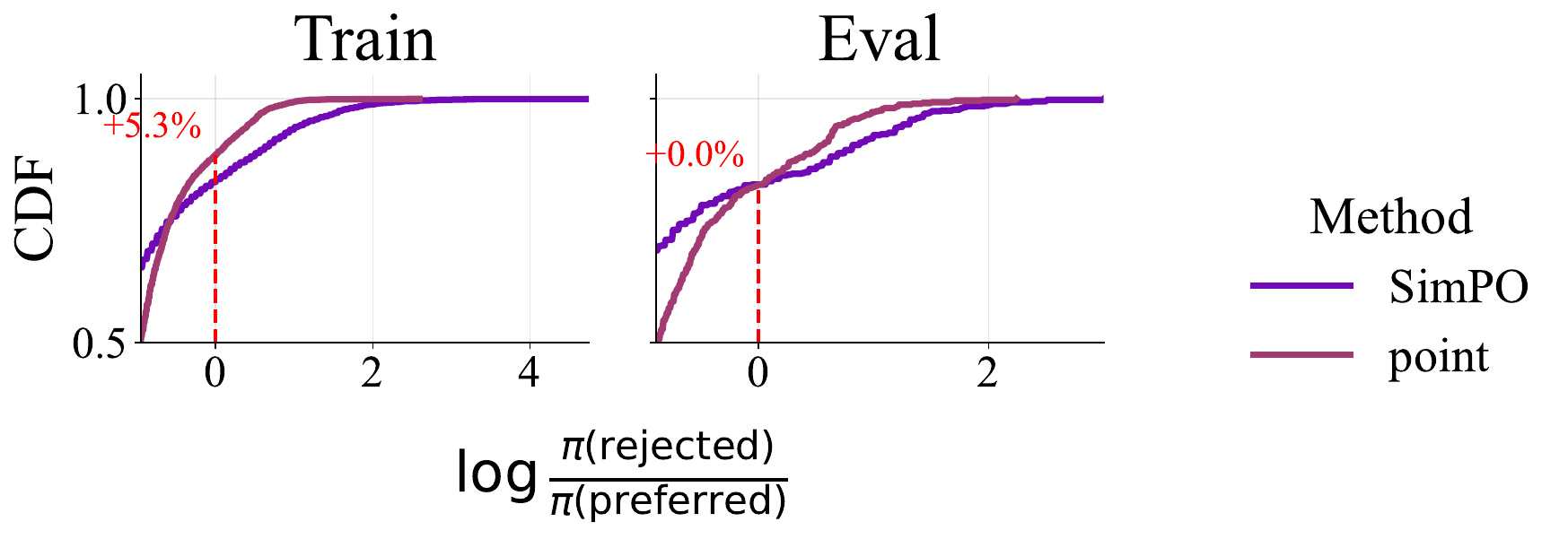}
   \caption{Empirical Cumulative Density Function (CDF) of likelihood ratios for preference data. In red we highlight the change in samples with negative ratios (correctly classified). \textit{Although pointwise constraints induce a  distribution shift in both train and test, it does not change the number of correctly classified preferences.}}
    \label{fig:orca-failure}
\end{figure}

\section{Ablations}\label{apx:ablations}

\subsection{Augmented Lagrangian outperforms the standard Lagrangian in practice.}\label{apx:aug_dual_vs_dual}

We compare the standard Lagrangian dual formulation in~\eqref{eq:lagrangian} with the augmented Lagrangian formulation in~\eqref{eq:aug_lagrangian_og}, using the same constraint tolerance $\epsilon=3$ and augmentation parameter $\alpha=1$ for the re-ranking task. The results in Figure~\ref{fig:aug-dual-vs-dual} show that the augmented formulation improves constraint satisfaction relative to the standard dual approach. This can be observed directly through the CVaR constraint violation, where the best augmented-dual configuration attains a value of $2.499$, improving over the best standard-dual value of $3.54$. However, the difference is more clearly reflected in the ranking performance, as measured by $\mathrm{MRR}@10$. The augmented formulation achieves the best overall $\mathrm{MRR}@10$ of $0.743$, while the best standard-dual configuration reaches only $0.711$. Importantly, this improvement does not come at the expense of the optimization objective: the augmented-dual configuration also attains a better performance with respect to length reduction. This joint improvement in relevance and objective performance is further reflected in the Pareto curve in Figure~\ref{fig:aug-dual-vs-dual}. 

The ablation also shows that the augmented Lagrangian is substantially more robust to the choice of dual step size. For the standard Lagrangian, decreasing the dual step size leads to a large degradation in $\mathrm{MRR}@10$, from $0.711$ at step size $1.0$ to $0.488$ at step size $10^{-3}$. In contrast, the augmented Lagrangian remains stable across step sizes $10^{-3}$, $10^{-2}$, and $10^{-1}$, with $\mathrm{MRR}@10$ values between $0.724$ and $0.743$. These results suggest that the quadratic augmentation makes the method less sensitive to the dual learning-rate hyperparameter. Nonetheless, fine-tuning the dual step size can still yield additional gains, with the best augmented-dual performance obtained around step size $10^{-1}$.

\begin{figure}[h]
\begin{subfigure}[c]{0.42\linewidth}
\centering
\includegraphics[width=\linewidth]{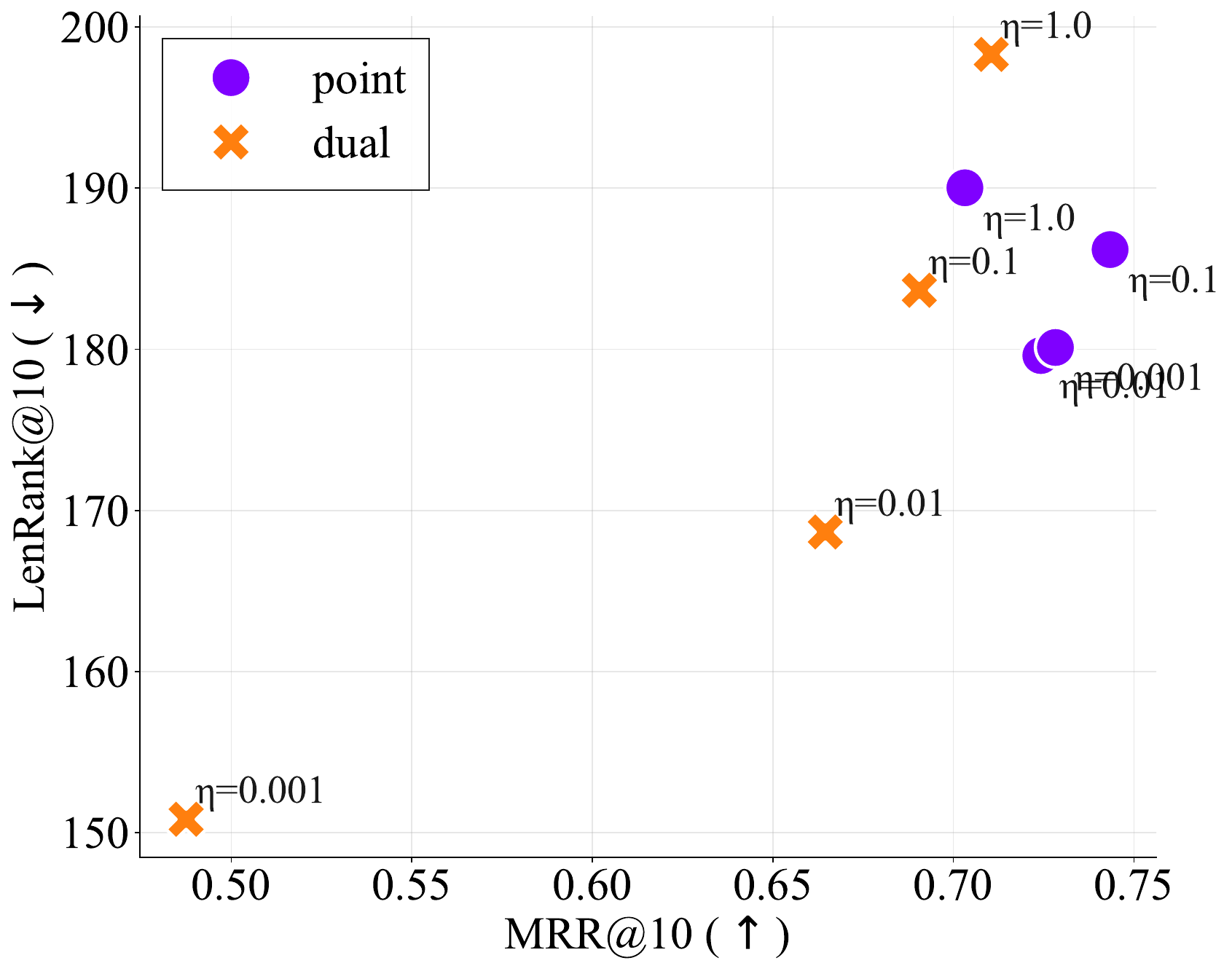}
\label{fig:aug_dual_vs_dual_pareto}
\end{subfigure}
\hfill
\begin{subfigure}[c]{0.56\linewidth}
\centering
\scriptsize
\resizebox{\linewidth}{!}{%
\begin{tabular}{llrrr}
\toprule
Method 
& $\eta$ 
& $\mathrm{MRR}@10$ 
& $\mathrm{LenRank}@10$ 
& $\mathrm{CVaR}_{95}$ \\
\midrule
\multirow{4}{*}{\centering Dual} 
& $0.001$ & $0.488$ & $150.840$ & $6.824$ \\
& $0.01$  & $0.665$ & $168.661$ & $4.964$ \\
& $0.1$   & $0.691$ & $183.652$ & $4.271$ \\
& $1.0$   & $0.711$ & $198.286$ & $3.549$ \\
\midrule
\multirow{4}{*}{\centering Aug. Dual} 
& $0.001$ & $0.728$ & $180.106$ & $3.092$ \\
& $0.01$  & $0.724$ & $179.606$ & $3.336$ \\
& $0.1$   & $0.743$ & $186.180$ & $2.499$ \\
& $1.0$   & $0.703$ & $190.012$ & $6.625$ \\
\bottomrule
\end{tabular}%
}
\end{subfigure}
\caption{Comparison between the standard dual formulation and the augmented dual formulation on the re-ranking task. All experiments use constraint tolerance $\epsilon=3$. For the augmented dual runs, we set the augmentation parameter to $\alpha=1$. \textit{Augmented Lagrangian consistently outperforms standard Lagrangian, and is more stable to the dual step size.}}
\label{fig:aug-dual-vs-dual}
\end{figure}

\subsection{Stability with respect to the augmentation parameter $\alpha$.} \label{apx:alpha_stab}
We also evaluate the sensitivity of the augmented dual formulation to the augmentation parameter $\alpha$. Table~\ref{tab:alpha-ablation} shows that performance remains stable across several orders of magnitude, from $\alpha=1$ to $\alpha=1000$. In particular, $\mathrm{MRR}@10$ varies only between $0.737$ and $0.743$, with mean $0.741$ and standard deviation $0.002$. Similarly, $\mathrm{Acc}@3$ remains essentially unchanged, while $\mathrm{LenRank}@10$ and the average top-3 length increase moderately as $\alpha$ grows. These results suggest that the augmented formulation is not highly sensitive to the precise choice of $\alpha$, although tuning $\alpha$ can still slightly improve the performance.

\begin{table}[h]
\caption{Augmentation-parameter ablation for the augmented dual formulation. Experiments were run on the re-ranking task using the pointwise method with constraint tolerance $\epsilon=3$ and dual step size $\eta=0.1$.}
\label{tab:alpha-ablation}
\centering
\small
\begin{tabular}{lrrrr}
\toprule
$\alpha$ 
& $\mathrm{Acc}@3$ 
& $\mathrm{MRR}@10$ 
& $\mathrm{LenRank}@10$ 
& Avg. Top-3 Length \\
\midrule
$1$    & $0.903$ & $0.743$ & $186.180$ & $67.025$ \\
$10$   & $0.894$ & $0.737$ & $192.073$ & $69.943$ \\
$100$  & $0.903$ & $0.742$ & $200.284$ & $73.085$ \\
$1000$ & $0.903$ & $0.742$ & $201.645$ & $73.323$ \\
\midrule
Mean $\pm$ Std. 
& $0.901 \pm 0.004$ 
& $0.741 \pm 0.002$ 
& $195.045 \pm 6.294$ 
& $70.844 \pm 2.577$ \\
\bottomrule
\end{tabular}
\end{table}

\subsection{Negligible runtime overhead.}

Table~\ref{tab:runtime-ablation} compares the time per epoch for the different methods for the re-ranking task. It shows that our methods do not introduce a noticeable runtime overhead compared with the standard training baselines. The time per epoch remains comparable across all methods, and the pointwise and relaxed variants are not slower in practice. This is expected, since the additional computation only requires accessing and updating the corresponding dual variable for each sample during training, which adds a negligible scalar operation relative to the cost of the model forward and backward passes.
\begin{table}[t]
\caption{Runtime comparison across methods for the reranking task. We report the mean and standard deviation of the time per training epoch, measured in seconds.}
\label{tab:runtime-ablation}
\centering
\small
\begin{tabular}{lr}
\toprule
Method & Time per epoch (s) \\
\midrule
MonoBERT & $942 \pm 003$ \\
avg & $939 \pm 036$ \\
pen & $931 \pm 030$ \\
point & $910 \pm 175$ \\
relaxed & $888 \pm 127$ \\
\bottomrule
\end{tabular}
\end{table}

\subsection{Negligible Memory Overhead} Although the pointwise method assigns one dual variable to each training sample, this does
not create a significant memory bottleneck. The multiplier vector
$\blambda\in\mathbb{R}^N_+$ need not be stored in GPU/TPU memory (VRAM), and can be accessed only for the
samples in the current minibatch. Therefore, the forward and backward passes
require only the minibatch entries of $\blambda$, not the full vector. Moreover,
the overhead is typically negligible relative to the dataset itself:
each sample multiplier is a single scalar, whereas each text example is represented by
a sequence of token indices, often together with attention masks, labels, and
other metadata.


\end{document}